\documentclass{amsart}
\usepackage[square,sort&compress,numbers]{natbib}
\usepackage[english]{babel}
\usepackage{hyperref}
\usepackage{amssymb,amsmath,amsfonts,amsthm}
\usepackage[capitalise]{cleveref}
\allowdisplaybreaks
\makeatletter
\def\@endtheorem{\endtrivlist}
\newtheorem*{rep@theorem}{\rep@title}
\newcommand{\newreptheorem}[2]{%
\newenvironment{rep#1}[1]{%
 \def\rep@title{#2 \ref{##1}}%
 \begin{rep@theorem}}%
 {\end{rep@theorem}}}
\makeatother
\theoremstyle{plain}
\newtheorem{thm}{Theorem}[section]
\newreptheorem{thm}{Theorem}
\newtheorem{lem}[thm]{Lemma}
\newtheorem{prp}[thm]{Proposition}
\newtheorem{cor}[thm]{Corollary}

\newtheorem{dfn}[thm]{Definition}
\theoremstyle{definition}
\newtheorem{rmk}[thm]{Remark}

\newtheoremstyle{note}
{3pt}
{3pt}
{\bfseries}
{\parindent}
{\bfseries\itshape}
{:}
{.5em}
{}
\theoremstyle{note}
\newtheorem*{Question}{Question}
\newtheorem*{Problem}{Problem}
\newtheorem*{Note}{Note}

\renewcommand{\eqref}[1]{\labelcref{#1}}
\crefname{thm}{Theorem}{Theorems}
\crefname{lem}{Lemma}{Lemmas}
\crefname{prp}{Proposition}{Propositions}
\crefname{cor}{Corollary}{Corollaries}
\crefname{dfn}{Definition}{Definitions}
\crefname{example}{Example}{Examples}
\crefname{rmk}{Remark}{Remarks}

\newcommand{\Ad}{\mathrm{Ad}}
\RequirePackage{slashed}
\newcommand{\sD}{\slashed{D}}

\RequirePackage{amssymb,amsmath}
\DeclareMathOperator{\Tr}{Tr}

\DeclareMathOperator{\ad}{ad}

\DeclareMathOperator{\Id}{Id}

\DeclareMathOperator{\End}{End}

\newcommand{\ol}[1]{\overline{#1}}

\title[Quantization of the cotangent bundle of a compact Lie group]{Quantization commutes with singular reduction:  cotangent bundles of compact Lie groups}
\author{Jord Boeijink, Klaas Landsman, and Walter  van Suijlekom}
\address{Institute for Mathematics, Astrophysics and Particle Physics, Faculty of Science, Radboud University Nijmegen, Heyendaalseweg 135, 6525AJ Nijmegen, The Netherlands}
\email{mail@jordboeijink.nl; landsman@math.ru.nl; waltervs@math.ru.nl}
\begin{document}
\begin{abstract}
We analyze the `quantization commutes with reduction' problem (first studied in physics by Dirac, and known in the mathematical  literature also as the \emph{Guillemin--Sternberg Conjecture})
 for the conjugate action of a compact connected Lie group $G$ on its own cotangent bundle $T^*G$.
  This example is interesting because the momentum map is not proper and the ensuing symplectic (or Marsden--Weinstein quotient) $T^*G /\!/ \Ad G$ 
   is typically singular.
   
  In the spirit of (modern) geometric quantization, our quantization of $T^*G$ (with its standard K\"ahler structure)  is defined as the kernel of  the Dolbeault--Dirac operator (or, equivalently, the spin$\mbox{}^{\mathbb{C}}$--Dirac operator) twisted by the pre-quantum line bundle. We show that this quantization of $T^*G$
  reproduces the Hilbert space found earlier by Hall (2002) using  geometric quantization based on a holomorphic polarization. 
  We then define the quantization of the singular quotient $T^*G /\!/ \Ad G$ as the kernel of the twisted Dolbeault--Dirac operator on the principal stratum, and show that quantization commutes with reduction in the sense that either way one obtains the same Hilbert space $L^2(\textbf{T})^{W(G,\textbf{T})}$. 
\end{abstract}
\maketitle

\vspace{\baselineskip}\noindent
\tableofcontents

\section{Introduction}\label{sct:introduction} 
Since this paper will soon become quite technical, we start with a conceptual introduction meant to provide the appropriate context for our mathematical results.

The two great fundamental theories of physics, namely General Relativity and the Standard Model (of elementary particle physics),
involve specific (gauge) field theories that are examples of constrained dynamical systems. This means that the variables (or degrees of freedom) in which the theory is \emph{a priori} defined are not the physical ones, both because there are (initial value) constraints on these variables and because 
some of them  are actually physically equivalent; this redundancy has to be removed by a reduction procedure that identifies such variables.
 These two aspects of constrained systems---i.e., constraints and reduction---turn out to be intimately linked, and both lead to problems in the quantization of such theories. These problems partly (but by no means only) arise  because the space of (duly constrained) physical degrees of freedom is typically \emph{singular}, i.e., not smooth. 

Constrained systems and their potential quantization were first systematically analyzed by Dirac \cite{Dirac}, whose work was subsequently rewritten in the language of symplectic geometry \cite{AM78,BSF,Sundermeyer} (we do not discuss the alternative algebraic treatment of constrained systems through the so-called BRST- or BV-formalisms here \cite{HT}). Although the theories just mentioned have infinitely many degrees of freedom (even after reduction), it is worth studying constrained systems with finitely many variables, both as an exercise for field theory and because these are of interest to mechanics in their own right
\cite{MarsdenRatiu,MarsdenMany}.

 Among finite-dynamical constrained systems, the case of Marsden--Weinstein reduction stands out  because of the clean definition  of the two steps of constraining and reducing  in terms of group actions on symplectic manifolds
 \cite{AM78, MarsdenRatiu,MW1974}. Indeed, let $M$ be a symplectic manifold and let some Lie group $G$ act on $M$ in strongly Hamiltonian fashion, so that it has an associated momentum map $j:M\rightarrow \mathfrak{g}^*$ (where $\mathfrak{g}^*$ is the dual of the Lie algebra $\mathfrak{g}$ of $G$). 
Then $j^{-1}(0)$ is the subspace of $M$ on which the constraints hold, and $M/\!/ G=j^{-1}(0)/G$, called the Marsden--Weinstein  (or symplectic)
 quotient (of $M$ by the given group action), is the reduced phase space in which all redundancies have been removed. Two regularity assumptions guarantee that $M/\!/ G$ is a manifold (which then is symplectic in a natural way): 
 \begin{enumerate}
\item  Zero must be a regular value of the momentum map $j$;
\item  The $G$-action on $j^{-1}(0)$ must be free. 
\end{enumerate}
In order to pose the `quantization commutes with reduction' problem, one needs:
\begin{itemize}
\item  Some quantization prescription $Q$ that 
maps $M$ and $M/\!/ G$ to certain `quantum data' $Q(M)$ and $Q(M/\!/ G)$  (involving Hilbert spaces and operators);
\item A `quantum reduction procedure' that transforms $Q(M)$ into quantum data $Q(M)/\!/ G$ by somehow mimicking classical Marsden--Weinstein reduction.
\end{itemize}
The  `quantization commutes with reduction'  problem is considered solved, then, if 
\begin{equation}
Q(M)/\!/ G\cong Q(M/\!/ G), \label{QMQMG}
\end{equation}
 where `$\cong$' denotes an appropriate  isomorphism whose nature (e.g.\ unitary) depends on the precise mathematical setting.
Such a solution was first achieved in the regular case by Guillemin and Sternberg \cite{GS82a,GS82b} for compact Lie group actions on compact K\"{a}hler manfolds, where the quantization procedure consisted of geometric quantization in the holomorphic polarization, the quantum data was a finite-dimensional Hilbert space $H(M)$ carrying a unitary $G$-action,  quantum reduction  $H(M)/\!/ G$ was defined by taking the $G$-invariant subspace of $H(M)$, and isomorphism just meant  equality of dimension of Hilbert spaces. This result  spurred a considerable mathematical literature in which `quantization commutes with reduction' was proved under increasingly relaxed assumptions (but still assuming regularity) and ensuing variations in the definition of quantization and quantum reduction; a sample of this literature, still under compactness assumptions, is \cite{Mei98,JefKir,TianZhang}.

The problem was subsequently generalized in two directions. First, one may allow the underlying spaces and groups to be non-compact, see, e.g., 
 \cite{Lan05,HL08,HM15, MZ09,MZ10,Paradan,Par11}, in all of which non-compactness was  tempered by requiring properness of the momentum map.
 Second, one may drop the regularity assumptions, so that singularities in $M/\!/ G$ may arise
 (typically maintaining compactness). 
The `quantization commutes with reduction' problem is much trickier in this case, if only because even the \emph{definition} of the quantization of the reduced space is at stake (more precisely: is even more ambiguous than it already was in the regular case).  Important guidance in this respect comes from  the fundamental paper \cite{SL91},  in which 
Sjamaar and Lerman proved that in the singular  case $M/\!/ G$ is stratified by symplectic manifolds, among which a `principal'  open dense stratum stands out;
see also the monographs \cite{OR04,Sny}.
On this basis, the first results on the problem more or less
in the tradition of the original work of Guillemin and Sternberg were given by Meinrenken and Sjamaar, who desingularized the reduced space 
\cite{MeinrenkenSjamaar};  see also \cite{Zhang,Teleman,Hue2006}. For rather different approaches, see  \cite{BCHS, HRS09,Hue11}, of which the latter two also analyze the quantization of $T^*G/\!/ \Ad G$.  See also \cite{LPS01}
for a survey of the field until 2000. 
In the present paper we perform a case study in which:
\begin{itemize}
\item  the original phase space $M=T^*G$ is non-compact (although $G$ is compact);
\item  the momentum map (defined by the pull-back of the adjoint $G$-action on itself) fails to be proper;
\item  the reduced phase space from these data is non-compact  as well as singular.
\end{itemize}
 Thus we have all possible kinds of trouble, albeit in a  relatively neat class of examples that can be completely worked out. As it turns out, reasonable notions of quantization and quantum reduction---still in the spirit of the original work of  Guillemin and Sternberg---exist also under these circumstances, and
  (\ref{QMQMG}) duly holds. \emph{Conceptually}, following \cite{Mei98} (who, in turn, as acknowledged in \cite{Sja96}, 
  followed an idea of Bott), we use an \emph{index-theoretic} definition of geometric quantization, as opposed to the original setting introduced by Kirillov, Kostant, and Souriau, as  reviewed in e.g. \cite{Hal2013,Woodhouse}. Our reasons for doing so lie precisely in the fact that the trouble alluded to above calls for a more flexible approach than the original one: see \cite{MeinrenkenSjamaar} for singular reduction and 
 \cite{HL08,Lan05,HM14,HM15,HS17a,HS17b} for the non-compact case. \emph{Technically}, for reasons explained before our key Definition \ref{dfn:quantization}, the interpretation of our notion of quantization (which \emph{a priori} is defined of the kernel of some operator) as an index can only be given \emph{a posteriori} from 
our Kodaira-style `vanishing' Theorem \ref{thm:kvt}.
 
 The price one pays for this increasingly abstract approach to quantization is a certain lack of feedback to the  physics problems that originally motivated the entire  `quantization commutes with reduction' enterprise, see above. Perhaps, however, these original problems were not formulated quite correctly  in the first place (the difficulty in quantizing General Relativity even heuristically, or Yang--Mills theory rigorously, seems to confirm this), their current mathematical formulation guiding the way to a future redefinition of quantization also in the context of physics. 
 
 We now explain our own approach in some detail, preceded by some conventions. 
 \subsection{Conventions}
Let $T^*M$ be the cotangent bundle of a smooth manifold $M$, and $\theta$ the canonical $1$-form on $T^*M$. The canonical symplectic structure on $T^*M$ is defined as  $\omega = d\theta$. 
The corresponding Liouville measure is
$
 \varepsilon =  \frac{(-1)^n}{n!} \omega^n,
$
provided the dimension of $M$ is equal to $n$.

If $J$ is an almost complex structure on $M$, then $J$ and $\omega$ are said to be \emph{compatible} if the symmetric $2$-form
\begin{align*}
 (X,Y) \mapsto \omega(JX,Y) =: g(X,Y), \quad (X,Y \in TM \times_M TM)
\end{align*}
defines a Riemannian metric on $M$.  The action of the almost complex structure on forms is defined as 
\begin{align*}
(J\alpha)(X) = -\alpha(JX), \quad (\alpha \in T_x^*M, X \in T_xM).
\end{align*}
Sesquilinear forms are always supposed to be anti-linear in the \emph{first} entry.
\subsection{Technical outline}
\label{sct:outline}
The canonical symplectic structure on the cotangent bundle $T^*G$ and the complex structure on $T^*G$ obtained by  identifying $T^*G$ with $G^\mathbb{C}$ (see \S\ref{sct:cot_bun} or \cite[Proof of Lemma 12]{Hal94}) combine into a K\"ahler structure \cite{Hal02}, which we we will refer to  as the \emph{standard K\"ahler structure} on $T^*G$. This is crucial for what follows. 
Our quantization of $T^*G$ is defined as the kernel of  the Dolbeault--Dirac operator on $T^*G$ \cite{Dui96,Fri00}
twisted by the pre-quantum line bundle \cite{Hal2013,Woodhouse}
\begin{align*}
(L,\nabla^L) := (T^*G \times \mathbb{C}, d + 2\pi i \theta),
\end{align*}
where $T^*G \times \mathbb{C}$ denotes the trivial hermitian complex line bundle over $T^*G$, and $\theta$ is the fundamental $1$-form on $T^*G$. Clearly, $(\nabla^L)^2 = 2\pi i \omega$, where $\omega = d\theta$ is the symplectic structure on $T^*G$. The Dolbeault--Dirac operator is defined as follows: 
\begin{dfn}
Let $(M,\omega)$ be a K\"ahler manifold. The \emph{Dolbeault--Dirac operator} on $\Gamma^\infty_c(M,\Lambda^{(0,\bullet)}T^*M)$ is the symmetric first-order differential operator given by
\begin{align*}
D= \sqrt{2}\left( \overline{\partial} + \overline{\partial}^* \right),
\end{align*} 
where the adjoint is taken with respect to the inner product 
\begin{align*}
 \langle s_1, s_2 \rangle = \int_M \langle s_1(x), s_2(x) \rangle \varepsilon, \quad (s_1,s_2 \in \Gamma^\infty_c(M,\Lambda^{(0,\bullet)}T^*M)).
\end{align*}
Here $\varepsilon$ denotes Liouville measure on $M$, and the hermitian structure on $\Lambda^{(0,\bullet)}T^*M$ is obtained by extending the Riemannian metric to a hermitian form on $T_\mathbb{C}M$, and then normalising it by dividing it by $k!$ on $\Lambda^{(0,k)} T^*M$ for each $0 \leq k \leq n$.

If $L$ is a hermitian line bundle with hermitian connection $\nabla^L$, then the \emph{twisted Dolbeault--Dirac operator} $D^L$ is defined as the symmetric first-order differential operator on $\Gamma^\infty_c(M,\Lambda^{(0,\bullet)}T^*M \otimes L)$ given by
\begin{align*}
 D^L =  \sqrt{2} \left( \overline{\partial}^L + (\overline{\partial}^L)^* \right),
\end{align*}
where $\ol{\partial}^L$ is the first-order differential operator on $\Lambda^{(0,\bullet)}T^*M \otimes L$ defined by
\begin{equation}
\ol{\partial}^L(\alpha \otimes s) := \ol{\partial}\alpha \otimes s + (-1)^{|\alpha|} \alpha \otimes (\nabla^L)^{(0,1)}s, \label{twist}
\end{equation}
where $\alpha$ is an $|\alpha|$-form and $s \in \Gamma^\infty(M,L)$.
\end{dfn}\begin{rmk}\label{5r}
\begin{enumerate}
\item Because the symplectic structure $\omega$ and the Riemannian metric $g$ are related through a complex structure, the Liouville measure defined by $\omega$ is equal to the Riemannian measure defined by $g$.
\item The hermitian connection $\nabla^L$ determines a holomorphic structure on $L$ as follows: A local section $s \in \Gamma^\infty(U,L)$ is holomorphic if and only if $\nabla^{(0,1)}s =0$ (\textit{cf.} \cite[Proposition 6.30]{GGK02}). With respect to this holomorphic structure,  $\nabla^L$ is the Chern connection. The untwisted or ordinary Dolbeault--Dirac operator is a special case of a twisted Dolbeault--Dirac operator, where $L$ is the trivial hermitian holomorphic vector bundle $L=M\times\mathbb{C}$.
\item In our case where $M$ is K\"ahler, the (twisted) Dolbeault--Dirac operator coincides with the (twisted) Spin$\mbox{}^{\mathbb{C}}$--Dirac operator,
 \textit{cf.}\  \cite[Prop.\ on p.\ 81]{Fri00} or \cite[Prop.\ 6.1]{Dui96}. For spin structures see  \emph{Note added in proof} at the end.
 \item Unless specified otherwise, the initial domain 
of a Dolbeault--Dirac operator (or of any other differential operator on a manifold $M$) is  taken to be $\Gamma^\infty_c(M,E)$, where $E$ is the (complex) vector bundle on whose sections the differential operator acts. 
If $M$ is geodesically complete, which by Theorem \ref{thm:TGgeodesicallycomplete} is the case for the unreduced space 
$M=T^*G$, then the (twisted) Dolbeault--Dirac operator is essentially self-adjoint on this initial domain by standard results \cite{Che73,Fri00,Wolf}, and hence has a unique self-adjoint extension.
The situation is more complicated for the  Dolbeault--Dirac operator on the reduced space  $T^*G /\!/\Ad G$, which is a 
 singular quotient whose principal stratum (on which our quantization will be defined) may not be geodesically complete. Fortunately,  in that case 
essential self-adjointness can be proved directly (albeit by a pretty elaborate argument), \textit{cf.}\ Proposition
\ref{prp:princstratess}.
\end{enumerate}\end{rmk}

The bundle $\Lambda^{(0,\bullet)}T^*M$ on $M$ decomposes into an even and an odd part as 
\begin{align*}
\Lambda^{(0,\bullet)}T^*M = \Lambda^{(0,even)} T^*M \oplus \Lambda^{(0,odd)} T^*M.
\end{align*}
With respect to this decomposition, the twisted Dolbeault--Dirac operator is an odd operator
\begin{align*}
 D^L = \left( \begin{array}{cc} 0 & D^L_-\\ D^L_+ & 0  \end{array} \right), 
\end{align*}
where $D_+$ maps $\Lambda^{(0,even)} T^*M$ into $\Lambda^{(0,odd)} T^*M$. If $M$ is compact, then its \emph{Dolbeault--Dirac quantization} is defined as
 \cite{Mei98}
$$\mathrm{index}(\ol{D}^L)=\dim(\mathrm{ker}(\ol{D}^L_+))-\dim(\mathrm{ker}(\ol{D}^L_-)),$$
where the bar denotes closure of the operator.
 On non-compact manifolds, however, the kernels of $\ol{D}^L_+$ and $\ol{D}^L_-$ may fail to be finite-dimensional, so that the naive Fredholm index cannot be defined. Indeed, the kernels of the twisted Dolbeault--Dirac operators on (non-compact) cotangent bundles of compact connect Lie groups studied in this paper \emph{are} infinite-dimensional, so this problem does arise. In these cases, one therefore needs to work with a different definition of Dolbeault--Dirac quantization. For the K\"ahler manifolds $M$ of interest in this paper, whose symplectic form is denoted by $\omega$  and whose associated  twisted Dolbeault--Dirac operator is called  $D^L$,  as before, encouraged by our later goal Theorem \ref{thm:qr0}
  we stipulate:
\begin{dfn}
\label{dfn:quantization}
Let $L\rightarrow M$ be a holomorphic prequantization line bundle on $M$, i.e., a 
hermitian holomorphic line bundle with Chern connection $\nabla^L$ such that $$(\nabla^L)^2 = 2\pi i \omega.$$ 
Then the \emph{Dolbeault--Dirac quantization} of $M$ is defined as the \emph{Hilbert space}
\begin{align}
\label{eq:DDq}
 \mathcal{Q}^L_{DD}(M) := \ker (\ol{D}_+^L).
\end{align}
\end{dfn}
\begin{rmk}\label{rmk:quantization}
  The kernel $\ker (\ol{D}_+^L)$ is an infinite-dimensional Hilbert space--that before reduction is equipped by construction with a unitary $G$-action--so that \eqref{eq:DDq} does not lead to an \emph{immediate} interpretation of quantization as an index. However, for those manifolds $M=T^*G$ and $M=T^*G/\!/ \Ad G$ whose Dolbeault--Dirac quantization we actually determine, the pertinent cokernel $\ker(\ol{D}_-^L)$ turns out to be trivial (\textit{cf.}\ Theorem \ref{thm:kvt}), so that in these cases
Definition \ref{dfn:quantization} is quite close to the rather more abstract  definition of quantization as an index, \textit{cf.}\ the Introduction. Thus we \emph{a priori} define quantization so as to result in a ($G$-) Hilbert space, whose conceptual index-theoretical abstraction only comes \emph{a posteriori}. See also  \S\ref{Outlook}.
\end{rmk}

The quantization of  the cotangent bundle $M=T^*G$ of a  compact connected Lie group $G$  with its standard K\"ahler structure was already studied by Hall \cite{Hal02}, whose Hilbert space was defined traditionally as the space of
all holomorphic sections of the trivial holomorphic (pre-quantum) line bundle with connection
\begin{equation}
(L,\nabla^L) := (T^*G \times \mathbb{C}, d + 2\pi i \theta).\label{linebundle}
\end{equation}
In that approach, the quantization of $T^*G$ is equal to the Hilbert space
\begin{align}
\label{eq:Hallq}
\mathcal{H}L^2(T^*G, e^{-2\pi|Y|^2}\varepsilon),
\end{align}
where $\varepsilon$ denotes the Liouville measure and where $(x,Y) \mapsto e^{-2\pi |Y|^2}$ is viewed as a function on $G \times \mathfrak{g} \cong T^*G$, the norm $|\cdot|$ coming from an $\Ad G$-invariant inner product on $\mathfrak{g}$ (see 
\S\ref{sct:cot_bun} for more details). Moreover, if the usual half-form correction of geometric quantization is taken into account, this Hilbert space is modified to 
\begin{align}
\label{eq:Hallqhalf}
\mathcal{H}L^2(T^*G, e^{-2\pi|Y|^2} \eta \varepsilon),
\end{align}
where $\eta$ is the $G\times G$-invariant function on $T^*G \cong G \times \mathfrak{g}$ determined by
\begin{align}
\label{eq:eta}
 \eta(Y) = \prod_{\alpha \in R^+} \frac{\sinh\left( \alpha(Y) \right)}{\alpha(Y)}, \quad (Y \in \mathfrak{t}).
\end{align}
Here $\mathfrak{t}$ is some maximal abelian subalgebra of $\mathfrak{g}$, and $R^+$ is a set of positive real roots. 
In his earlier work Hall had already constructed explicit unitary isomorphisms between the Hilbert spaces $\mathcal{H}L^2(T^*G, e^{-2\pi|Y|^2}\varepsilon)$ and $L^2(G)$  \cite[Theorem 10]{Hal94}, as well as  between $\mathcal{H}L^2(T^*G, e^{-2\pi|Y|^2} \eta \varepsilon)$ and $L^2(G)$ \cite[Theorem 2.6]{Hal02}). The second isomorphism is just a constant times the inverse Segal--Bargmann transform for $G$ \cite{Hal97a}. The first isomorphism is 
written down explicitly in the appendix of the present paper, where we show that it is $G \times G$-equivariant when $L^2(G)$ is endowed with the natural $G\times G$-action.
Either way, there are natural identifications of the above quantizations with $L^2(G)$. Towards our final result (i.e.\ Theorem \ref{thm:qr0}),  an important intermediate step lies in the following connection to Hall's work:
 \begin{repthm}{thm:quantspin}
 Let $G$ be a compact connected Lie group and endow $T^*G$ with its standard K\"ahler structure. Let $(L,\nabla^L)$ be the asociated 
 $G \times G$-equivariant pre-quantization line bundle (\ref{linebundle}). 
  Then the Dolbeault--Dirac quantization  (\ref{eq:DDq}) of $T^*G$ is 
   $G \times G$-equivariantly 
  equal to Hall's Hilbert space (\ref{eq:Hallq}),
equipped with the natural $G \times G$-action. 
Consequently, by  \cite[Theorem 2.6]{Hal02}, our quantization of $T^*G$ is $G \times G$-equivariantly isomorphic to $L^2(G)$.
\end{repthm}
 
  Along the way, we prove that the canonical line bundle $K$ on $T^*G$ is semi-negative, which is an interesting result on its own:
\begin{repthm}{thm:can_semineg}
 The canonical line bundle $K$ on $T^*G$ (with its standard K\"ahler structure) is semi-negative. 
\end{repthm}
Recall that a hermitian holomorphic line bundle $L$ over a K\"{a}hler manifold $M$ is called \emph{semi-negative} if the matrix $R_{ij}$ in the expansion
$R=\sum_{i,j} R_{ij}dz^id\overline{z}^j$ is semi-negative (in that $-R_{ij}$ is positive semidefinite). 

In \S\ref{sct:quantization_cot_max_torus} we deal with the Dolbeault--Dirac quantization of the \emph{singular} Marsden-Weinstein quotient 
\begin{equation}
T^*G /\!/\Ad G := j^{-1}(0) /\Ad G, \label{defMWQ}
\end{equation}
where $j$ is the momentum map as defined in Lemma \eqref{lem:moment_map}.
Although the above Marsden-Weinstein quotient is not a smooth manifold, it is still a symplectic stratified space \cite{Sja90,SL91} (see also \cite{OR04} for an extensive account on symplectic stratified spaces). In particular, there exists a (unique) \emph{principal stratum}, which is an open and dense subset of $T^*G /\!/ \Ad G$ carrying a natural symplectic structure and, in our case, even a (natural) K\"ahler structure. The third main result of this paper is:

\begin{repthm}{thm:princstrateq}
If the Dolbeault--Dirac quantization of $j^{-1}(0) /\Ad G$ is defined to be the Dolbeault--Dirac quantization of its principal stratum, then
\begin{align*}
 \mathcal{Q}_{DD}(j^{-1}(0) /\Ad G) \cong \mathcal{Q}_{DD}(T^*\textbf{T})^{W(G,\textbf{T})}.
\end{align*}
Here $\textbf{T}$ denotes a maximal torus in $G$ and $W(G,\textbf{T})=N_G(\textbf{T})/\textbf{T}$ is the associated  Weyl group.
\end{repthm}
It might seem awkward at first glance to ignore the singular strata, so let us elaborate on this a bit. As will be shown in Lemma \ref{lem:mwred}, the Marsden-Weinstein quotient $T^*G /\!/\Ad G$ is homeomorphic to $T^*\textbf{T} / W(G,\textbf{T})$. The pre-image of the principal stratum of $T^*G /\!/ \Ad G$ under the projection map $T^* \textbf{T} \rightarrow T^*\textbf{T} / W(G,\textbf{T})$ is an open and dense submanifold of $T^*\textbf{T}$. Because the stratification is by symplectic (and hence even-dimensional) manifolds, pre-images of the \emph{singular strata} all have codimension at least $2$ in $\textbf{T} \times \mathfrak{t}$.
From this, one can show that each compactly supported section of $\Lambda^{(0,\bullet)}T^*(T^*\textbf{T}) \otimes L$ can be approximated in the graph norm of $D^L$ by sections with compact support in the pre-image of the principal stratum. Therefore, the closure of the Dolbeault--Dirac operator on $L^2(T^*\textbf{T})$ is `insensitive' to the  removal of the pre-images of the singular strata, and, consequently, it is sufficient to consider the Dolbeault--Dirac operator on the pre-image of the principal stratum. By discreteness of the Weyl group, it is then also sufficient to quantize the principal stratum of the singular Marsden-Weinstein quotient. See \S \ref{sct:rq} for details.

On the other hand, the Weyl integration formula and Theorem \ref{thm:quantspin} imply that the quantum reduction of $L^2(G)$ at zero, which is defined as the $G$-invariant part of the Hilbert space $\mathcal{Q}^L_{DD}(T^*G) \cong L^2(G)$ \cite{GS82b,Lan05}, is isomorphic to $L^2(\textbf{T})^{W(G,\textbf{T})}$. We thus arrive at the main conclusion of this paper:
\begin{repthm}{thm:qr0}On Definition \ref{dfn:quantization} of Dolbeault--Dirac quantization of $T^*G$ and $T^*G /\!/ \Ad G$, quantization after reduction and reduction after quantization are both canonically isomorphic to $L^2(\textbf{T})^{W(G,\textbf{T})}$.
\end{repthm}
The isomorphisms in question, as well as their `canonical' nature, will be explained in due course. 
\section{The K\"ahler structure on the cotangent bundle}
\label{sct:cot_bun}
Let $G$ be a compact connected Lie group and consider its cotangent bundle $T^*G$. The following considerations show that $T^*G$ is canonically a K\"ahler manifold \cite{Hal94}.

 Using left-trivialisation we can identify  the cotangent bundle $T^*G$ with $G \times \mathfrak{g}^*$, where $\mathfrak{g}$ denotes the Lie algebra of $G$. Fix an $\Ad G$-invariant inner product $\langle \cdot, \cdot \rangle$ on $\mathfrak{g}$, which always exists by compactness of the Lie group $G$, and use this inner product to identify $\mathfrak{g}$ with its dual $\mathfrak{g}^*$, and
hence  $G \times \mathfrak{g}^*$ with $G \times \mathfrak{g}$. Using left-trivialisation we always identify the tangent spaces $T_{(g,Y)}(G \times \mathfrak{g})$ and cotangent spaces $T^*_{(g,Y)}(G \times \mathfrak{g})$, where $(g,Y) \in G \times \mathfrak{g}$, with $\mathfrak{g} \times \mathfrak{g}$ and $\mathfrak{g}^* \times \mathfrak{g}^*$, respectively.

For given $G$ as above, there exists a  connected complex Lie group $G^\mathbb{C}$ such that  every homomorphism of $G$ into a complex Lie group $H$ extends to a holomorphic homomorphism from $G^\mathbb{C}$ into $H$, and this group is unique (up to isomorphism of Lie groups) provided it
it contains $G$ as a closed subgroup and has
Lie algebra $\mathfrak{g} \oplus i \mathfrak{g}$  \cite[\S3]{Hal94}. The associated complex structure on $G^\mathbb{C}$ is determined by the property that for each $x\in G$ and $Y\in \mathfrak{g}$, the map $z\mapsto xe^{zY}$ is holomorphic from $\mathbb{C}$ to $G^\mathbb{C}$. 
Identifying the Lie algebra 
$\mathfrak{g} \oplus i \mathfrak{g}$ of $G^\mathbb{C}$  with
$\mathfrak{g} \times \mathfrak{g}\equiv \mathfrak{g}^2$ (in the same way that $\mathbb{R}\oplus i\mathbb{R}=\mathbb{C}$ is identified with $\mathbb{R}\times \mathbb{R}$),
the map $J_e: \mathfrak{g}^2\rightarrow \mathfrak{g}^2$ defining this complex structure by translation is then simply given by the familiar expression
\begin{equation}
J_e= \left( \begin{array}{cc} 0 & -1 \\ 1 & 0 \end{array} \right). \label{Je}
\end{equation}
Using left translation, for each point $t \in G^\mathbb{C}$, the tangent space $T_t G^\mathbb{C}$ may be identified with $\mathfrak{g} \times \mathfrak{g}$, where the first copy of $\mathfrak{g}$ consists of the vectors  tangent to $G \subset G^\mathbb{C}$.
One obtains  a diffeomorphism $G^\mathbb{C}\cong T^*G$ through the
 polar decomposition 
\begin{eqnarray*}
\Phi: G \times \mathfrak{g} &\rightarrow& G^\mathbb{C};\\
 (x,Y)  &\mapsto& xe^{iY};
\end{eqnarray*}
see  \cite[Proof of Lemma 12]{Hal94}. This endows $T^*G$ and $G \times \mathfrak{g}$ with a complex structure, too, in which for each
$(x,Y)\in  G \times \mathfrak{g}$ the map $z\mapsto (xe^{uY},vY)$ is holomorphic (where $z=u+iv)$. In turn, the canonical symplectic structure on $T^*G$ transfers to a symplectic structure on $G^\mathbb{C}$,  and the two combine to a K\"ahler structure \cite{Hal97,Hal02}  on the latter manifold, and hence also on
$T^*G$ and on  $G \times \mathfrak{g}$.
We will often tacitly  identify $T^*G$, $G \times \mathfrak{g}$, and $G^\mathbb{C}$, referring to the above K\"ahler structure as \emph{standard}. On $G \times \mathfrak{g}$, the standard K\"ahler structure
 has a global K\"ahler potential given by the $G\times G$-invariant (real-valued) function $ (g,Y) \mapsto |Y|^2$, i.e.,
\begin{align*}
 \omega = -i \partial \ol{\partial} |Y|^2,
\end{align*}
where $|\cdot|$ denotes the norm of $Y \in \mathfrak{g}$ w.r.t.\ the $\Ad G$-invariant inner product  \cite{Hal02}.

The cotangent bundle $T^*G$ carries a natural $G \times G$-action coming from the $G\times G$-action on $G$ given by left and (inverse) right multiplication. In what follows, we will use  some explicit formulas summarised in the following lemma.
\begin{lem}
\label{lem:formulas}
The $G \times G$-actions on $G \times \mathfrak{g}$ and on $G^\mathbb{C}$ are given by
\begin{align*}
 (h_1, h_2) \cdot (x,Y) = (h_1xh^{-1}_2,\Ad_{h_2}Y), \quad (h_1, h_2)\cdot  t   = h_1 t h_2^{-1}, 
\end{align*}   
respectively, where $x\in G$, $h_{1,2} \in G$,  $Y\in \mathfrak{g}$, and $t\in G^\mathbb{C}$. The fundamental $1$-form $\theta$ on $G \times \mathfrak{g}$ is equal to
\begin{align*}
 \theta_{(x,Y)} (X_1, X_2) =   \langle Y, X_1 \rangle,
\end{align*}
and the symplectic structure on $G \times \mathfrak{g}$ (inherited from $T^*G$) is
\begin{align}
\label{eq:symp}
\omega_{(x,Y)} ( (X_1, X_2) , (Z_1, Z_2) ) = \langle X_2, Z_1 \rangle_\mathfrak{g} - \langle X_1, Z_2 \rangle_\mathfrak{g} - \langle Y, [X_1,Z_1] \rangle_\mathfrak{g}.
\end{align}

If we identify the tangent spaces $T_{(x,Y)} (G \times \mathfrak{g})$  and $T_t G^\mathbb{C}$  with $\mathfrak{g} \times \mathfrak{g}$ as above, then 
 as an isomorphism of $\mathfrak{g} \times \mathfrak{g}$, the differential $T\Phi_{(x,Y)}: T_{(x,Y)} (G \times \mathfrak{g}) \rightarrow T_{xe^{iY}} G^\mathbb{C}$, is
\begin{align}
\label{eq:differentialPhi}
 T_{(x,Y)}\Phi = \left( \begin{array}{cc} \cos \ad Y & \frac{1 - \cos \ad Y}{\ad Y} \\ -\sin \ad Y & \frac{\sin \ad Y}{\ad Y}\end{array} \right).
\end{align}
\end{lem}
\begin{proof}
Eq.\  (\ref{eq:differentialPhi}) is proved in \cite{Hal97}. The other equations are straightforward.
\end{proof}

Let $\{e_k\}$ be an orthonormal basis of $\mathfrak{g}$ for the $\Ad G$-invariant inner product $\langle \cdot,\cdot \rangle_\mathfrak{g}$. All forms on $G \times \mathfrak{g}$ are $C^\infty(M)$-linear combinations of the left-invariant forms $\{\alpha_k\}$, where $\alpha_k(e_G) = e^*_k$ in $\mathfrak{g}^*$, and the forms $\{dy_k\}$, where $(y_k)_k$ are the coordinates of $Y \in \mathfrak{g}$ with respect to the orthonormal basis $\{e_1, \dots e_n \}$ of $\mathfrak{g}$. Similarly, we choose left-invariant $1$-forms $\{\eta_k\}$ on $G^\mathbb{C}$ such that $\eta_k(e_{G^\mathbb{C}}) = (e^*_k,0)$ for all $k \in \{1,\dots n\}$. Note that our definition of $J$ on forms is  such that $J\eta_k(e_{G^\mathbb{C}}) = (0,e_k^*)$.

We later need the fact that the K\"ahler structure on $T^*G$ is geodesically complete:
\begin{thm}
\label{thm:TGgeodesicallycomplete}
Let $G$ be a compact connected Lie group and endow $T^*G$ with the standard K\"ahler structure and write $g$ for the corresponding Riemannian metric. The Riemannian manifold $(T^*G,g)$ is geodesically complete.  
\end{thm}
\begin{proof}
 By the Hopf-Rinow theorem \cite{Ber02}, a Riemannian manifold is geodesically complete if and only if there exists a real-valued function $f$ such that $\|df\|_g \leq C$ for some $C >0$ and $M_c= \{ x \in M \mid f(x) < c \}$ is relatively compact for each $c \in \mathbb{R}$. In our case, take the function $$f: (x,Y) \mapsto \log (1 + |Y|^2)$$ on $G \times \mathfrak{g}$, where $|Y|$ is the norm of $Y \in \mathfrak{g}$ with respect to the chosen Ad-invariant inner product on $\mathfrak{g}$. First of all, it is clear that
\begin{align*}
 \log (1 + |Y|^2) < c \iff |Y|^2 < e^c -1,
\end{align*}
for all $c \in \mathbb{R}$. So, $M_c = \{(x,Y) \in G \times \mathfrak{g} \mid |Y|^2 < e^c - 1 \}$ is relatively compact for all $c \in \mathbb{R}$. We now show that $\|df\|_g \leq 2$. We have
\begin{equation}
\|df\|^2_g= g(df,df) = g(df^\#, df^\#) = \omega(J(df)^\#, df^\#), \label{begindf}
\end{equation}
 where $(df)^\#$ is the unique vector field that satisfies $df(W) = g((df)^\#, W)$ for all vector fields $W$. Let $\alpha_k,dy_k$, $\eta_k$ and $J\eta_k$ be as above. Since $d(\psi \circ s)(m) = \psi'(s(m)) ds(m)$ for any $s \in C^\infty(M,\mathbb{R})$ and $\psi: \mathbb{R} \rightarrow \mathbb{R}$, we obtain
\begin{equation}
 df = d(\log(1 + |Y|^2)) = \frac{1}{1 + |Y|^2} d(1+|Y|^2) = \frac{1}{1 + |Y|^2} \sum_{k=1}^n 2y_kdy_k. \label{df2}
\end{equation}  
 We next calculate $J df$. To this end, note that $(1+|Y|^2)df = \sum_{k=1}^n 2y_kdy_k$ corresponds (at the point $(x,Y)$ in $G\times \mathfrak{g}$) to $(0,2Y)$ in the basis $\{\alpha_k,dy^k\}$. From (\ref{Je}), the complex structure $J$ on $T^*_{(x,Y)} (G \times \mathfrak{g})$ is
\begin{align*}
(T_{(x,Y)}\Phi)^* \left( \begin{array}{cc} 0 & -1 \\ 1 & 0 \end{array} \right)   (T_{(\Phi(x),\Phi(Y))}\Phi^{-1})^*.
\end{align*}
This can be computed from (\ref{eq:differentialPhi}). The computation is easy, as
$\ad Y(Y)=0$, so that 
\begin{equation}
(T\Phi)^* \left( \begin{array}{cc} 0 & -1 \\ 1 & 0 \end{array} \right)   (T\Phi^{-1})^*  \left( \begin{array}{c} 0\\ 2Y\end{array}
 \right)=
\left( \begin{array}{c} -2Y \\ 0\end{array} \right). \label{df33}
\end{equation}

Let us write $Z:=Z(x,Y) =(Z_1(x,Y),Z_2(x,Y)) =: (Z_1,Z_2)$ for a vector field on $G \times \mathfrak{g}$. On the one hand, we have
\begin{align}
\label{eq:covectorvector}
 (1 + |Y|^2)(Jdf)(Z) = -2\langle Y,Z_1 \rangle_\mathfrak{g}.
\end{align}
On the other hand, 
\begin{align}
\label{eq:metric}
 (1 + |Y|^2) g((Jdf)^\#, Z) &= (1 + |Y|^2) \omega (J(Jdf)^\#, Z)\nonumber \\ &= - (1 + |Y|^2) \omega (df^\#,Z).
\end{align}
Using (see Lemma \ref{lem:formulas}) 
\begin{align*}
 \omega_{(x,Y)}((X_1,X_2),(Z_1,Z_2)) = \langle X_2, Z_1 \rangle_{\mathfrak{g}} - \langle X_1 , Z_2 \rangle_\mathfrak{g} - \langle Y, [X_1,Z_1] \rangle_\mathfrak{g}
\end{align*}
and equating \eqref{eq:covectorvector,eq:metric}, we obtain
\begin{equation}
 2\langle Y,Z_1 \rangle_\mathfrak{g} = (1 + |Y|^2) \left( \langle (df)^\#_2, Z_1 \rangle_{\mathfrak{g}} - \langle(df)^\#_1 , Z_2 \rangle_\mathfrak{g} - \langle Y, [(df)^\#_1,Z_1]\rangle_\mathfrak{g} \right).
\end{equation}
This equality holds for all $Z$ if and only if $(df)^\#(x,Y) =\left (0,\frac{2Y}{1+|Y|^2} \right)$. An argument similar to 
(\ref{df2}) - (\ref{df33})  shows that $$(1+|Y|^2)J(df)^\# = (-2Y,0).$$ Hence from (\ref{begindf}) we obtain
\begin{align*}
\|df\|^2_g= \omega(J(df)^\#, df^\#) = (1+|Y|^2)^{-2} \omega((-2Y,0),(0,2Y))=  (1+|Y|^2)^{-2} 4|Y|^2.
\end{align*}
Hence $\|df\|^2_g\leq 4$, so $\|df\|^2_g$
 is bounded. Since we also showed that $M_c$ is relatively compact for all $c \in \mathbb{R}$, we have proved that $(T^* G,g)$ is geodesically complete.
\end{proof}

\section{Quantization of the cotangent bundle of a compact connected Lie group}
\label{sct:quantization}
In this section we prove one of the two main results of this paper, in that we recover the equivariant Hilbert space \eqref{eq:Hallq},
originally found by Hall  \cite{Hal02}, from our definition (\ref{eq:DDq}) of `geometric' quantization. We do this by showing that the kernel of the Dolbeault--Dirac operator does not contain any forms of non-zero degree, so that the kernel consists of all (square-integrable) holomorphic \emph{sections} of $L$. We are then precisely in the setting of Hall's papers mentioned in the Introduction, and we can apply his (Bargmann-type) isomorphisms to (unitarily) identify the Dolbeault--Dirac quantization with $L^2(G)$.

In order to prove that the kernel of the Dolbeault--Dirac operator does not contain forms of non-zero degree, we need a non-compact version of the Kodaira vanishing theorem. In \S\ref{sct:Kodaira} we prove a version of this theorem for geodesically  complete K\"ahler manifolds.  In \S\ref{sct:can_semi-neg} we prove that the canonical line bundle on $T^*G$ is semi-negative (see Definition \ref{dfn:positivebundle}). The proof we give there is a corrected form of an argument due to Bielawski \cite{Bie03}. 
Using the non-compact version of the Kodaira vanishing theorem, in \S\ref{sct:DD_spin_quantization} we deduce that the kernel of the Dolbeault--Dirac operator does not contain any forms of non-zero degree. 
 
\subsection{The Kodaira vanishing theorem for complete K\"ahler manifolds}
\label{sct:Kodaira}
The main result of this section is an extension of Kodaira's vanishing argument on compact K\"ahler manifolds (see e.g \cite{BGV92}) to non-compact, geodesically complete, K\"ahler manifolds. More specifically, we show that if $K^* \otimes L$ is a positive line bundle over $M$ (see \S\ref{dfn:positivebundle} below), where $K$ denotes the canonical line bundle and $L$ denotes a hermitian holomorphic line bundle, then the kernel of the closure $\ol{D}^L$ of the twisted Dolbeault--Dirac $D^L$ operator on $M$ is contained in the smooth $(0,0)$-forms, i.e., the smooth sections of $L$. 
The relation between our results and other generalizations of  Kodaira's vanishing theorem to the non-compact case, such as \cite[Thm.\ 2.8]{Ohsawa} and \cite[Thm.\ 3.5.15]{MaMa} (provided by the referee) is unclear to us, but it looks remote. 

We first collect some facts about unbounded operators $T:\mathcal{D}(T)\rightarrow H$, where $H$ is some Hilbert space and 
$\mathcal{D}(T)\subset H$ is the domain of $T$, always assumed dense in $H$. We call $T$
  \emph{essentially positive} if it is essentially self-adjoint and satisfies 
 \begin{align}
\label{eq:pos}
 \langle x,Tx \rangle \geq 0, \:\: x \in \mathcal{D}(T).
\end{align}
Similarly, $T$ is called \emph{positive} if it is self-adjoint and satisfies \eqref{eq:pos}.
\begin{prp}\label{prop1to4}
\begin{enumerate}
\item If $T$ is essentially positive, then its closure $\ol{T}$ is positive.
\item If $T$ is positive,  there exists a unique positive operator $T^\frac{1}{2}$ with $(T^{\frac{1}{2}})^2 = T$. 
\item  If $S$ and  $T$ are essentially positive operators, defined on a common dense domain $\mathcal{D} \subset \mathcal{H}$, then 
the operator $C:= \ol{S + T}$ has the following properties:
\begin{eqnarray}
\mathcal{D}(C)& \subset& \mathcal{D}(\ol{S}^\frac{1}{2}) \cap \mathcal{D}(\ol{T}^{\frac{1}{2}});\label{DCclaim}\\
 \langle Cx,x \rangle& =& \langle\ol{S}^\frac{1}{2}x,\ol{S}^\frac{1}{2}x \rangle + \langle\ol{T}^\frac{1}{2}x,\ol{T}^\frac{1}{2} x\rangle, \:\: x \in \mathcal{D}(C); \label{eq:positivesum}\\
 \ker C& \subset &  \ker \ol{S} \cap\ker\ol{T}. \label{cor:sumofpositiveoperators}
\end{eqnarray}
\end{enumerate}
\end{prp}
\begin{proof}
Part (1) is routine, and part (2) is  \cite[Proposition 5.13]{Sch12}).
For part (3), let $(x_n)_{n=1}^\infty$ be a sequence  in $\mathcal{D}$ such that $x_n \rightarrow x$ and $C x_n \rightarrow z =:Cx$ in $\mathcal{H}$. We show that $(\ol{S}^\frac{1}{2} x_n)_n$ is a Cauchy sequence in $\mathcal{H}$. First,
\begin{align*}
 \langle Cy,y \rangle = \langle Sy,y \rangle + \langle Ty,y \rangle = \langle \ol{S}^\frac{1}{2}y,\ol{S}^\frac{1}{2}y \rangle + \langle \ol{T}^\frac{1}{2}y,\ol{T}^\frac{1}{2} y \rangle \geq \langle \ol{S}^\frac{1}{2}y,\ol{S}^\frac{1}{2}y \rangle ,
\end{align*}
for all $y \in \mathcal{D}$. In particular,
\begin{align*}
\langle \ol{S}^\frac{1}{2}(x_n-x_m),\ol{S}^\frac{1}{2}(x_n -x_m ) \rangle  &\leq \langle C(x_n-x_m),x_n -x_m \rangle \\&\leq \|Cx_n - Cx_m\| \|x_n - x_m\|,
\end{align*}
and the right-hand side goes to zero as $n,m$ go to infinity, since both $(x_n)_n$ and $(Cx_n)_n$ are Cauchy sequences. In particular, $(\ol{S}^\frac{1}{2}x_n)_n$ is a Cauchy sequence, and hence $x \in \mathcal{D}(\ol{S}^\frac{1}{2})$ with $\ol{S}^\frac{1}{2}x = \lim_n \ol{S}^\frac{1}{2} x_n$. Similarly, $x \in \mathcal{D}(\ol{T}^\frac{1}{2})$, whence (\ref{DCclaim}). 

To prove (\ref{eq:positivesum}), we already know that it holds on $\mathcal{D}$. Let $x \in \mathcal{D}(C)$ be as above. By the previous argument,
 $\ol{S}^\frac{1}{2} x_n \rightarrow \ol{S}^\frac{1}{2} x$ and $\ol{T}^\frac{1}{2} x_n \rightarrow \ol{T}x$. Thus 
\begin{align*}
\langle Cx,x \rangle &= \lim_n \langle Cx_n,x_n \rangle = \lim_n \left(\langle \ol{S}^\frac{1}{2}x_n,\ol{S}^\frac{1}{2}x_n \rangle + \langle \ol{T}^\frac{1}{2}x_n,\ol{T}^\frac{1}{2} x_n\rangle \right) \\&= \langle \ol{S}^\frac{1}{2}x,\ol{S}^\frac{1}{2}x\rangle + \langle\ol{T}^\frac{1}{2}x,\ol{T}^\frac{1}{2}x\rangle,
\end{align*}
which proves (\ref{eq:positivesum}) for each $x \in \mathcal{D}(C)$.

Finally, if  $x \in \ker C$, then $x \in \mathcal{D}(\ol{S}^{\frac{1}{2}}) \cap \mathcal{D}(\ol{T}^\frac{1}{2})$ by (\ref{DCclaim}),
 and 
\begin{align*}
\langle Cx,x \rangle = \langle \ol{S}^\frac{1}{2}x,\ol{S}^\frac{1}{2}x \rangle + \langle \ol{T}^\frac{1}{2}x,\ol{T}^\frac{1}{2} x\rangle = \| \ol{S}^\frac{1}{2}x \|^2 + \|\ol{T}^\frac{1}{2}x\|^2 = 0.
\end{align*}
Therefore, $x \in \ker \ol{S}^\frac{1}{2} \cap \ol{T}^\frac{1}{2} = \ker\ol{S} \cap \ker \ol{T}$.  This gives (\ref{cor:sumofpositiveoperators}).  
\end{proof}

Let $V$ be a \emph{finite}-dimensional inner product space over $\mathbb{K}=\mathbb{C}$. 
A self-adjoint operator $A: V \rightarrow V$ is called \emph{positive-definite} if $\langle v,Av \rangle >0$ for all non-zero $v \in V$. This is equivalent to $A$ being a positive, invertible operator on $V$.
The following lemma is just a matter of linear algebra. We omit the proof. 
\begin{lem}
\label{lem:positivetensoralgebra}
 Let $A$ be a positive-definite (and hence self-adjoint) linear operator on a finite-dimensional inner product space $(V,\langle,\rangle)$ and let $\widetilde{A}$ be its extension to $TV$ as a derivation, i.e.,
\begin{align*}
 \widetilde{A}(v \otimes w) = Av \otimes w + v \otimes Aw, \quad (v,w \in V),
\end{align*}
and 
\begin{align*}
 \widetilde{A}z = 0, \quad (z \in \mathbb{C} = T^0 V). 
\end{align*}
Let $W$ be another vector space and consider the linear map
\begin{align*}
 \widetilde{A} \otimes \Id: TV \otimes W \rightarrow TV \otimes W.
\end{align*}
Then $\widetilde{A} \otimes \Id$ restricts to a positive-definite self-adjoint operator on $T^kV \otimes W$ for each $k \geq 1$. Consequently, $\ker \widetilde{A} \otimes \Id = T^0V \otimes W$.
\end{lem}

We now turn to (possibly unbounded) zeroth-order differential operators. The next proposition shows that any symmetric zeroth-order differential operator is essentially self-adjoint.

\begin{prp}
\label{prp:esa_zeroorder}
Let $E \rightarrow M$ be a hermitian vector bundle over an arbitrary oriented Riemannian manifold $M$. Consider the Hilbert space $L^2(M,E)$, where the measure on $M$ is the Riemannian measure. If $R$ is a smooth vector bundle homomorphism such that $R_x\in \End(E_x)$ is symmetric for each $x \in M$, then $R$ is essentially self-adjoint on the domain $\Gamma_c^\infty(M,E)$. Its closure $\ol{R}$ is positive if $R_x$ acts fibrewise by positive operators.
\end{prp}
\begin{proof}
Since $R_x$ is symmetric for each $x \in M$, $R$ is a symmetric zeroth-order differential operator on $\Gamma^\infty_c(M,E)$. We prove that $\ol{R} = R^*$ as an operator on $L^2(M,E)$. 
The domain of $R^*$ is equal to
\begin{align*}
 \mathcal{D}(R^*) =\{ s \in L^2(M,E) \mid Rs \in L^2(M,E) \}, 
\end{align*}
and $R^*s = Rs$ on this domain, where $R$ acts on $s$ as a zeroth-order differential operator.  Indeed, since $R$ is a zeroth-order differential operator, the equality
\begin{align*}
\langle s, Rt \rangle = \int_M \langle s(x), R(x)t(x) \rangle_x dx = \int_M \langle R(x)s(x),t(x) \rangle_x  dx,
\end{align*}
holds for all $s \in L^2(M,E)$, $t \in \mathcal{D}(R)$. In particular, $s \in \mathcal{D}(R^*)$ if and only if $x \mapsto R_x s_x \in L^2(M,E)$ and then $R^*s = Rs$.

If $s \in \mathcal{D}(R^*)$, then 
\begin{align*}
 \langle s,R^* t \rangle = \langle s, Rt \rangle = \langle Rs, t \rangle \leq \|Rs\|_{L^2} \|t\|_{L^2},
\end{align*}
for all $t \in \mathcal{D}(R^*)$. Here, $R$ is again viewed as a zeroth-order differential operator. Consequently, $s \in \mathcal{D}(R^{**})$, and so $R^* \subset R^{**} = \ol{R}$. Because $R$ is symmetric, $\ol{R} \subset R^*$, and hence $\ol{R} = R^*$. In other words, $R$ is essentially self-adjoint.
\end{proof}

To formulate the Kodaira vanishing theorem, we need the notion of a positive line bundle (\textit{cf.} \cite{BGV92}).
\begin{dfn}
\label{dfn:positivebundle}
 A hermitian holomorphic line bundle $L$ over $M$ is said to be \emph{(semi-)positive} if its curvature $R^L$ is of the form
\begin{align*}
R^L = \sum_{k,l} R_{kl} dz^k \wedge d\ol{z}^l,
\end{align*}
where $(R_{kl})$, which is always hermitian, is a positive (semi-)definite matrix at each point. A hermitian holomorphic line bundle $L$ is said to be \emph{(semi-)negative} if $L^*$, which has curvature $-R^L$, is (semi-)positive. 
\end{dfn}

\begin{rmk}
\label{rmk:arbitraryframe}
Note that positive (semi-)definiteness of the matrix $(R_{kl})$ at a point $x$ is equivalent to saying that the sesquilinear form $(v_1,v_2) \mapsto \ol{R_x(v_1,\ol{v}_2)}$ defines a positive (semi-)definite hermitian form on $T^{(1,0)}M$. Hence, equivalently, $L$ is a (semi-)positive line bundle if $R^L(x)$ defines a positive (semi-)definite hermitian form on $T_x^{(1,0)}M$ for each point $x \in M$. This condition on $R^L$ can be checked point-wise and, moreover, it can be checked with respect to an arbitrary frame of $T_x^{(1,0)}M$. In particular, Definition \ref{dfn:positivebundle} makes sense. 
\end{rmk}

In view of the above remark we therefore say that $R^L$ is \emph{(semi-)positive at $x$} if $R^L(x)$ defines a positive (semi-)definite hermitian form on $T_x^{(1,0)}M$, and we denote this by $R^L(x) > 0$ ($R^L(x) \geq 0$). If $R^L(x)$ is positive (semi-)definite at each point $x$, then $R^L$ is simply said to be \emph{(semi-)positive}, and this is denoted by $R > 0$ ($R \geq 0$). Similar definitions are introduced for the (semi-)negative case.

We now state a Kodaira vanishing theorem for geodesically complete K\"ahler manifolds. The corresponding result for the compact case can be found in \cite[Proposition 3.72]{BGV92} and \cite[Proposition 6.1]{Dui96}. The theorem states that the kernel of $\ol{D}^L$ contains no forms of non-zero degree. On compact manifolds this is equivalent to the vanishing of the higher cohomology groups of $L$.

\begin{thm}[Kodaira vanishing theorem]
\label{thm:kvt}
Let $M$ be a geodesically complete K\"ahler manifold.  If $K^* \otimes L$ is a positive line bundle, then 
\begin{align*}
\ker \ol{D}^L \subset \Gamma^\infty(M,L),
\end{align*}
where $D^L$ denotes the twisted Dolbeault--Dirac operator. In other words, $\ker \ol{D}^L$ is concentrated in degree $0$, and in fact
$\ker(\ol{D}_-^L)=\{0\}$.
\end{thm}
\begin{proof}
Since $D^L$ is an elliptic operator on $M$, the kernel of $\ol{D}^L$ is contained in the smooth sections of $\Lambda^{(0,\bullet)} T^*M \otimes L$ (see for instance \cite[Proposition 10.4.8]{HR00}). We show that all the sections in the kernel are necessarily $(0,0)$-forms. To do so, we first recall the \emph{Bochner-Kodaira formula} (see for instance \cite[Proposition 3.71]{BGV92}), which says that on a K\"ahler manifold the square of $D^L$ on $\Gamma^\infty_c(M,\Lambda^{(0,\bullet)} T^*M \otimes L)$ is equal to
\begin{align*}
\tfrac{1}{2}  (D^L)^2 = \Delta^{(0,\bullet)} + \sum_{k,l} e( \ol{\xi}_k) i(\ol{Z}_j) R^{K^*\otimes L}(Z_j,\ol{Z}_k),  
\end{align*}
where $\{\ol{\xi}_k\}$ is any unitary frame of $T^{*(0,1)}M$ with dual frame $\{\ol{Z}_k\}$ in $T^{(0,1)}M$. Denote the second (zeroth-order differential) operator by $R$. As $M$ is geodesically complete, the operators $\Delta^{(0,\bullet)}$ and $(D^L)^2$ are essentially self-adjoint on the domain $\Gamma_c^\infty(M,\Lambda^{(0,\bullet)} T^*M \otimes L)$. Since 
\begin{align*}
\int_M \langle \Delta^{(0,\bullet)} s,s \rangle dx  = \int_M \langle \nabla^{(0,1)}s, \nabla^{(0,1)} s \rangle dx, \quad s \in \Omega_c^{(0,\bullet)}(M,L),
\end{align*}
the closure of $\Delta^{(0,\bullet)}$ is a positive self-adjoint operator by Proposition \ref{prop1to4}.1.

Note that $R$ acts trivially on $L$, so that we can simply regard $R$ as a morphism of the bundle $\Lambda^{(0,\bullet)}T^*M$. Since $K^* \otimes L$ is positive, the bundle endomorphism 
\begin{align*}
R=\sum_{k,l} e( \ol{\xi}_k) i(\ol{Z}_j) R^{K^*\otimes L}(Z_j,\ol{Z}_k)
\end{align*}
 acts by invertible, positive operators on the fibres of $T^{*(0,1)}M$. Indeed, with respect to the unitary frame $\{\ol{\xi}_k\}$, the matrix of $R$ on $T^{*(0,1)}M$ is precisely $R(Z_j,\ol{Z}_k)$, which is positive-definite by assumption. If $\omega_1, \omega_2$ are forms of degree $(0,|\omega_1|)$ and $(0,|\omega_2|)$, respectively, the action of $R$ on $\omega_1 \wedge \omega_2$ is
\begin{align*}
 R(\omega_1 \wedge \omega_2) &=\sum_{j,k} R^{K^* \otimes L}(Z_j,\ol{Z}_k) e(\ol{\xi}_k) \left(  i(\ol{Z}_j) \omega_1 \wedge \omega_2 + (-1)^{|\omega_1|} \omega_1 \wedge i(\ol{Z}_j) \omega_2 \right) \\ &=  \sum_{j,k} R^{K^* \otimes L}(Z_j,\ol{Z}_k) \left( e(\ol{\xi}_k) i(\ol{Z}_j) \omega_1 \wedge \omega_2 +  \omega_1 \wedge e(\ol{\xi}_k) i(\ol{Z}_j) \omega_2 \right) \\ &= R\omega_1 \wedge \omega_2 + \omega_1 \wedge R\omega_2.
\end{align*}
This means that $R$ acts (fibrewise) as a derivation of degree $0$. We can now apply Lemma \ref{lem:positivetensoralgebra}: let $A$ be the restriction of $R_x$ to $V=T_x^{*(0,1)}M$. Then, with the notation of Lemma \ref{lem:positivetensoralgebra}, $R_x$ is the restriction of $\widetilde{A}$ on $T(T_x^{*(0,1)}M)$ to $\Lambda^\bullet(T_x^{*(0,1)}M)$. Hence $R$ acts as a positive operator on the fibres of $\Lambda^\bullet(T^{*(0,1)}M) \otimes L$. By 
Proposition \ref{prp:esa_zeroorder} the operator $R$ is essentially positive.

The operators $(\ol{D}^L)^2$ and $\ol{(D^L)^2}$ are both self-adjoint extensions of the essentially self-adjoint operator $(D^L)^2$, hence they are equal. Thus  $s \in \ker \ol{D}^L$ if and only if $s \in \ker \ol{(D^L)^2}$. Since both $\Delta^{(0,\bullet)}$ and $R$ are essentially positive on $\Gamma_c^\infty(M,\Lambda^\bullet(T^{*(0,1)}M) \otimes L)$, an element $s \in \ker \ol{D}^L \subset  \Gamma^\infty(M,\Lambda^\bullet(T^{*(0,1)} M) \otimes L)$ is also in $\ker R$ by
 (\ref{cor:sumofpositiveoperators}), i.e., $Rs =0$. As $s$ is smooth, so is $Rs$. Thus  $Rs = 0$ if and only if $R_x s(x) = 0$ for all $x \in M$.  Another application of Lemma \eqref{lem:positivetensoralgebra} shows that $s(x) \in (\Lambda^0 T^{*(0,1)}M \otimes L)_x$ for each $x$. Consequently, $s \in \Gamma^\infty(M,L)$.
\end{proof}

\subsection{Semi-negativity of the canonical line bundle}
\label{sct:can_semi-neg}
We prove that the canonical line bundle on $T^*G$ is semi-negative for any compact connected Lie group. Let $\{W^k\}$ be linearly independent left-invariant holomorphic $(1,0)$-forms on $G^\mathbb{C}$, and write
\begin{align*}
 \omega = -i \sum_{k,l} g_{kl} W^k \wedge \ol{W}^l.
\end{align*}
The Liouville measure is defined as
\begin{align*}
 \varepsilon &= \frac{(-1)^n}{n!} \omega^n = \frac{i^n}{n!}  \left(\sum_{k,l} g_{kl} W^k \wedge \ol{W}^l\right)^n \\&= \frac{i^n}{n!} \left( n! \sum_{\sigma \in S_n} (-1)^\sigma g_{1 \sigma(1)} \cdots g_{n \sigma(n)} W^1 \wedge \ol{W}^1 \wedge \cdots \wedge  W^n \wedge \ol{W}^n \right) \\
&= i^n (-1)^{n(n+1)/2} \det (g_{kl}) \ol{W}^1 \wedge \cdots \wedge \ol{W}^n \wedge W^1 \wedge \cdots \wedge W^n.
\end{align*}
By \cite{Hal02}, the function $\eta$ of \eqref{eq:eta} satisfies 
\begin{align*}
b c^2  \eta^2 \varepsilon = \ol{W}^1 \wedge \cdots \wedge \ol{W}^n \wedge W^1 \wedge \cdots \wedge W^n,
\end{align*}
where $b =i^n(-1)^{n(n-1)/2}$ and $c>0$ a constant, so that
\begin{align*}
\det (g_{kl}) =  \frac{1}{ i^n (-1)^{n(n+1)/2} bc^2 \eta^2} = \frac{1}{ i^{2n} (-1)^{n(n+1)/2} (-1)^{n(n-1)/2}c^2 \eta^2} = \frac{1}{c^2\eta^2}.
\end{align*}
In particular, the curvarture of the canonical line bundle $K$ is
\begin{align*}
 \partial\overline{\partial} \log \det (g_{kl}) = 2 \overline{\partial}\partial \log \eta.
\end{align*}
Consequently, the canonical line bundle $K$ is semi-negative if and only if $\eta$ is plurisubharmonic. The key point in proving semi-negativity of the canonical line bundle is the following proposition due to Bielawsky \cite{Bie03}.\footnote{Sign errors that propagate throughout Bielawsky's proof made us (and perhaps other readers) feel increasingly insecure. Hence we give an independent discussion with a complete proof. See also \cite[\S XIII.4]{Neebbook} (provided by the referee) for analogous results (not including the one above).}

Let $\textbf{T}$ be a maximal torus in the compact connected group $G$, $\textbf{t}$ its Lie algebra, and $W(G,\textbf{T})$ the corresponding Weyl group. Any $G \times G$-invariant function on $G \times \mathfrak{g}$ (and hence on $T^*G\cong G \times \mathfrak{g}$ or on $G^\mathbb{C}\cong T^*G$) is determined by its values on $\{e \} \times \mathfrak{t}$ (as this subset intersects every orbit), and this restriction is  $W(G,\textbf{T})$-invariant. 
Conversely, $W(G,\textbf{T})$-invariant functions on $\mathfrak{t}$ uniquely extend to $G \times G$-invariant functions on $G \times \mathfrak{g}$.
In what follows, we also use the non-trivial fact that the bijective correspondence between these function spaces preserves smoothness \cite{Schwarz} (see also \cite[Cor.\ 5.11]{Helgason}).
\begin{prp}
\label{prp:con_psh}
The restriction of $G \times G$-invariant functions on $G^\mathbb{C}$ to $\{e\} \times \mathfrak{t}$ determines a bijective correspondence between smooth $G \times G$-invariant plurisubharmonic functions on $G^\mathbb{C}$ and smooth $W(G,\textbf{T})$-invariant convex functions on $\mathfrak{t}$.
\end{prp}

To prove this, let us introduce some notation (which is taken from \cite{Bie03} as well). Reagard $G \times \mathfrak{g}$ as a left $G$-principal bundle, where $G$ acts on $G \times \mathfrak{g}$ by left-multiplication on $G$. If $Y \in \mathfrak{g}$, write $Y^*$ for the fundamental vector field on $G$ associated to $Y$. Because $G$ acts from the left, the vector field $Y$ is right-invariant and $[X^*,Y^*] = - [X,Y]^*$ for all $X,Y \in \mathfrak{g}$. 
Write $J$ for the complex structure on $G \times \mathfrak{g}$. One can define a connection $1$-form on $G \times \mathfrak{g}$ by
\begin{align*}
 \theta(X^* + JY^*) = X, \quad (X \in \mathfrak{g}),
\end{align*}
and a $\mathfrak{g}$-valued $1$-form $L$, vanishing on vertical vectors, by $L = J\theta$. Suppose now that we are given a smooth $W(G,\textbf{T})$-invariant function $\widetilde{f}$ on $\mathfrak{t}$ extended uniquely to a $G \times G$-equivariant smooth function $f$ on $G^\mathbb{C}$ (we do not assume yet that $\widetilde{f}$ is convex). We write $\langle \cdot, \cdot \rangle_\mathfrak{g}$ for a (fixed) $\Ad G$-invariant inner product on $\mathfrak{g}$ and extend it to a complex-bilinear form on $\mathfrak{g}_\mathbb{C}$.
Because $J$ is left $G$-invariant, the definition of the Lie bracket of vector fields implies that 
\begin{align*}
[X^*,JY] = J[X^*,Y], \quad \text{for $X \in \mathfrak{g}$ and $Y$ a vector field}.
\end{align*}
Then, by the Nijenhuis-Schouten bracket,
\begin{align}
\label{eq:JJvert}
 [X^*,Y^*] =  - [JX^*,JY^*], \quad \text{for all $X,Y \in \mathfrak{g}$}.
\end{align}
\begin{lem}
 On $\{e\} \times \mathfrak{g}$, we have
\begin{align*}
i (\partial f - \overline{\partial} f)  = \langle \nabla f, \theta \rangle_\mathfrak{g}.
\end{align*}
\end{lem}
\begin{proof}
 We first evaluate both sides on vectors of the form $X - iJX$ and $X + i JX$, ($X \in \mathfrak{g}$):
\begin{align*}
 i(\partial f - \overline{\partial} f)(X-iJX) = i (\partial f)(X - iJX) = i df(X - iJX) =df(JX),
\end{align*}
where, in the final step, we used the fact that $f$ only depends on $\mathfrak{g}$. From \eqref{eq:differentialPhi} and its inverse one can deduce that $$JX = \left(\frac{1- \cos \ad Y}{\sin \ad Y}X,\frac{\ad Y}{\sin \ad Y} X\right).$$  Therefore,\footnote{Although we work with right-invariant vector fields here, \eqref{eq:differentialPhi}, which was formulated for the left-trivialisation of $TG$, can still be applied. The reason is that we compute in the point $\{e\} \times \mathfrak{g}$.}
\begin{align*}
 df(JX) = \left\langle \nabla f, \frac{\ad Y}{\sin \ad Y} X \right\rangle_\mathfrak{g} = \left\langle \frac{\ad Y}{\sin \ad Y} \nabla f,  X \right\rangle_\mathfrak{g}.
\end{align*}
Note that $[Y,\nabla f] =0$, because $f$ is $\Ad G$-invariant. Expanding $\frac{\ad Y}{\sin \ad Y}$ in a power series, we see that
\begin{align*}
 i(\partial f - \overline{\partial} f)(X-iJX) = \langle \nabla f,  X \rangle_\mathfrak{g}= \langle \nabla f,  \theta \rangle_\mathfrak{g}(X-iJX).
\end{align*}
Similarly, for each $X \in \mathfrak{g}$,
$i(\partial f - \overline{\partial} f)(X + iJX) = \langle \nabla f,  \theta \rangle_\mathfrak{g}(X+iJX)$.
\end{proof}

The form $\partial f - \overline{\partial} f$ is invariant under the left-action of $G$. However, for the form $\langle \nabla f, \theta \rangle_\mathfrak{g}$, initially defined on $\{e\} \times \mathfrak{g}$, to be invariant under left-multiplication, we need to extend $\nabla f$ to a $G$-equivariant map from $G \times \mathfrak{g}$ to $\mathfrak{g}$, i.e., 
\begin{align*}
(\nabla f)_{(g,Y)} = \Ad_g (\nabla f)_{(e,Y)}, \quad \text{for all $g \in G$, $Y \in \mathfrak{g}$}. 
\end{align*}
Following \cite{Bie03}, we call this map $\mu$. Thus  by definition,
\begin{align}
\label{eq:potentialK}
i(\partial f - \overline{\partial} f) = \langle \mu, \theta \rangle_\mathfrak{g}
\end{align}
on $G \times \mathfrak{g}$. The map $\mu$ is $G$-equivariant, so the covariant derivative of $\mu$ (with respect to the connection $1$-form $\theta$) $D\mu = d\mu - \theta \wedge \mu$ vanishes on vertical vector fields. At each point $(g,Y) \in G \times \mathfrak{g}$,  the map $L$ determines an isomorphism from the horizontal tangent space at that point to $\mathfrak{g}$. Still following \cite{Bie03}, for each $(g,Y) \in G\times\mathfrak{g}$ we define
\begin{align*}
\Theta(g,Y) := D\mu(g,Y) \circ L^{-1}(g,Y): \mathfrak{g} \rightarrow \mathfrak{g}. 
\end{align*}
Taking the exterior derivative on both sides of \eqref{eq:potentialK}, we obtain
\begin{align*}
2i \overline{\partial} \partial f = d\langle \mu, \theta \rangle_\mathfrak{g} .
\end{align*}
Write $[\alpha,\alpha](X,Y) := [\theta(X), \theta(Y)]$ for any $\mathfrak{g}$-valued $1$-form $\alpha$ and vector fields $X,Y$. Then, by \eqref{eq:JJvert},
\begin{align*}
(d\theta  - [\theta, \theta])(JX^*, JY^*) = -\theta([JX^*,JY^*]) = - [X,Y] = -[L,L](JX^*, JY^*). 
\end{align*} 
Because both $d\theta - [\theta,\theta]$ and $[L,L]$ vanish on vertical vector fields, it follows that 
\begin{align*}
 d\theta - [\theta, \theta ] = - [L,L].
\end{align*}
\begin{lem}
\label{lem:psh_theta}
 The $G \times G$-invariant function $f$ is plurisubharmonic if and only if the hermitian map $\Theta - i \ad \mu: \mathfrak{g}_\mathbb{C} \rightarrow \mathfrak{g}_\mathbb{C}$ is positive.
\end{lem}
\begin{proof}
Let us first rewrite the expression for $i \overline{\partial} \partial f = d\langle \mu, \theta \rangle_\mathfrak{g}$:
\begin{align*}
 d\langle \mu, \theta \rangle_\mathfrak{g}&= \langle d\mu , \theta \rangle_\mathfrak{g} + \langle \mu ,d\theta \rangle_\mathfrak{g} = \langle D\mu ,\theta \rangle_\mathfrak{g} + \langle \theta \wedge \mu,   \theta \rangle_\mathfrak{g}  + \langle \mu , [\theta, \theta] \rangle_\mathfrak{g} - \langle \mu, [L,L] \rangle_\mathfrak{g} \\ &= \langle D\mu ,\theta \rangle_\mathfrak{g} - \langle \mu , [\theta, \theta] \rangle_\mathfrak{g} - \langle \mu, [L,L] \rangle_\mathfrak{g},   
\end{align*}
where $$\langle \alpha,\beta \rangle_\mathfrak{g}(X,Y) = \langle \alpha(X), \beta(Y) \rangle_\mathfrak{g} - \langle \alpha(Y), \beta(X) \rangle_\mathfrak{g}$$ for all vector fields $X,Y$ and $\mathfrak{g}$-valued $1$-forms $\alpha, \beta$. Note that, as a $(1,1)$-form, $d\langle \mu, \theta \rangle_\mathfrak{g}$ is $J$-invariant, i.e., $$d\langle \mu, \theta \rangle_\mathfrak{g}(JX,JY) =  d\langle \mu, \theta \rangle_\mathfrak{g}(X,Y)$$ for all vector fields $X,Y$, so that $\Phi: \mathfrak{g} \rightarrow \mathfrak{g}$ is symmetric. We are now able to evaluate $\overline{\partial}\partial f$ on pairs of vector fields. We have:\footnote{In \cite[on top of pp. 129]{Bie03} one can find a similar result, in which in our opinion $\Theta + i \ad \mu$ should be replaced by $-\Theta + i\ad \mu$ (in \cite{Bie03}, the symbol $\Phi$ is used instead of $\Theta$).} 
\begin{align*}
& \overline{\partial}\partial f(X^* - i JX^*  , Y^*+ iJY^* )  \\ &=- \frac{1}{2}\langle D\mu(JX^*) ,Y \rangle_\mathfrak{g} - \frac{1}{2}\langle D\mu(JY^*) ,X \rangle_\mathfrak{g} + \frac{i}{2} \langle \mu , [X,Y] \rangle_\mathfrak{g} + \frac{i}{2}\langle \mu, [X,Y] \rangle_\mathfrak{g} \\ &=- \frac{1}{2}\langle\Theta(X) ,Y) \rangle_\mathfrak{g} - \frac{1}{2}\langle \Theta(Y) ,X \rangle_\mathfrak{g} +  i\langle \mu , [X,Y] \rangle_\mathfrak{g} = -\langle \Theta(X) ,Y \rangle_\mathfrak{g} +  i\langle \mu , [X,Y] \rangle_\mathfrak{g} \\ &
=- \langle  Y, \Theta(X) \rangle_\mathfrak{g} + i \langle Y, [\mu,X]  \rangle_\mathfrak{g}
= - \langle Y, (\Theta - i \ad \mu )X \rangle_\mathfrak{g}.
\end{align*}
Write $h$ for the sesquilinear form $h(Z,W) = \overline{\partial}\partial f(Z,\overline{W})$ (which is the only time in this paper we take the sesquilinear form to be anti-linear in the \emph{second} variable) and write $\langle \cdot ,\cdot \rangle_ \mathfrak{g}^s$ for the extension of the inner product on $\mathfrak{g}$ to an inner product on $\mathfrak{g}_\mathbb{C}$, \emph{anti-linear} in the first variable. Then we find
\begin{align*}
h(X^* - i JX^*  , Y^*- iJY^* ) = - \langle Y,(\Theta - i\ad \mu)X \rangle_\mathfrak{g}^s, 
\end{align*}
for all $X,Y \in \mathfrak{g}_\mathbb{C}$. From this expression it is clear that $f$ is plurisubharmonic if and only if the hermitian map $\Theta - i\ad \mu: \mathfrak{g}_\mathbb{C} \rightarrow \mathfrak{g}_\mathbb{C}$ is positive.
\end{proof}

Let $R$ denote the set of \emph{real} roots of $(G,\textbf{T})$, i.e., $R = \{ \alpha \mid  i\alpha \text{ is a root for } (G,\textbf{T}) \}$.
\begin{lem}
\label{lem:theta_eigenvalues}
 The eigenvalues of the hermitian map $\Theta - i\ad \mu$ on $\{e\} \times \mathfrak{t}$ are given by the eigenvalues of the Hessian matrix of $\widetilde{f}$ and by 
 \begin{equation}
\alpha(\mu(Y)) \left( \frac{\cosh(\alpha(Y))}{\sinh(\alpha(Y))} + 1  \right)\:\:\: (\alpha \in R).\label{almuY}
\end{equation}
\end{lem}
 If $\alpha(Y) =0$, this expression should be interpreted as a limit: rewrite the expression as $\frac{\alpha(\mu(Y))}{\alpha(Y)} \left( \frac{\alpha(Y)\cosh(\alpha(Y))}{\sinh(\alpha(Y))} + \alpha(Y) \right)$. The second factor is defined if  $\alpha(Y)=0$, because $(t / \sinh t)|_{t=0}= (\sinh t / t)^{-1}|_{t=0}=1$. To see what the first factor means, choose an orthonormal basis $\{e_i\}_{i=1}^k$ for $\mathfrak{t}$ such that $\{ e_2, \dots, e_k\}$ spans $\ker \alpha$. Let $(y_i)$ be the corresponding coordinates for $Y \in \mathfrak{t}$. Then $\alpha(\mu(Y)) / \alpha(Y) = \frac{1}{y_1}\frac{\partial \widetilde{f}}{\partial y_1}$. Now, $\lim_{y_1 \rightarrow 0} \frac{1}{y_1}\frac{\partial \widetilde{f}}{\partial y_1}(y_1, \dots y_k) =  \frac{\partial^2 \widetilde{f}}{\partial y_1^2}(0,y_2, \dots, y_k)$, where we have used $\frac{\partial \widetilde{f}}{\partial y_1}(0,y_2, \dots y_k) =0$ for all $y_2, \dots y_k$ by $W(G,\textbf{T})$-invariance of $\widetilde{f}$.

\begin{proof}
Decompose $\mathfrak{g}_\mathbb{C}$ as 
\begin{align*}
 \mathfrak{g}_\mathbb{C} = \mathfrak{t}_\mathbb{C} \oplus \bigoplus_{i \alpha} \mathfrak{g}_{i\alpha}, 
\end{align*}
where $\mathfrak{g}_{i\alpha}$ is the root space belonging to the root $i\alpha$. Consider the composition
\begin{align*}
(\Theta - i\ad \mu) \circ L,
\end{align*}
where $L$ is considered a map $\mathfrak{g} \rightarrow \mathfrak{g}$, because on $\{e\} \times \mathfrak{g}$, the map 
\begin{align*}
L:T_{(e,Y)} (G\times \mathfrak{g}) = \mathfrak{g} \times \mathfrak{g} \rightarrow \mathfrak{g}
\end{align*}
 is given by
\begin{align*}
 L (X_1,X_2) = \frac{\sin \ad Y}{\ad Y} \cdot X_2.
\end{align*}
Now,
\begin{align*}
(\Theta - i\ad \mu) &= (D\mu - i \ad \mu \circ L) L^{-1} = (d\mu + [\mu, \theta - iL]) L^{-1} \\ &=  \left( d\mu + \left[\mu, \frac{e^{-i \ad Y} -1}{\ad Y}\right] \right) \circ L^{-1}.
\end{align*}
Here, $\left[\mu, \frac{e^{-i \ad Y} -1}{\ad Y}\right](X) := \left[\mu, \frac{e^{-i \ad Y} -1}{\ad Y} X\right]$, $(X \in \mathfrak{g}$). Since $\mu$ maps $\mathfrak{t}$ into itself, $d\mu$ maps $\mathfrak{t}_\mathbb{C}$ into $\mathfrak{t}_\mathbb{C}$ and $\left[\mu, \frac{e^{-i \ad Y} -1}{\ad Y}\right]$ acts trivially on $\mathfrak{t}_\mathbb{C}$. As $L^{-1}$ is the identity map on $\mathfrak{t}_\mathbb{C}$, the map $\Theta - i\ad \mu$ maps $\mathfrak{t}_\mathbb{C}$ into itself and this map is given by $d\mu|_{\mathfrak{t}_\mathbb{C}}$.

Now, $\mu =\nabla f= \sum_i \frac{\partial f}{\partial x_i} e_i$, seen as a $\mathfrak{g}$-valued function on $\mathfrak{g}$, where $(x_i)_i$ are the coordinates of $\mathfrak{g}$ with respect to some basis $\{e_i\}$ of $\mathfrak{g}$. Hence, $d\mu(\frac{\partial}{\partial x_j}) =  \frac{\partial^2 f}{\partial x_i \partial{x_j}}e_i$. Consequently, the eigenvalues of $d\mu|_{\mathfrak{t}}$ are given by the eigenvalues of the Hessian matrix of $\widetilde{f}$ on $\mathfrak{t}$.
On $X_{i\alpha}  \in \mathfrak{g}_\mathfrak{i\alpha}$, the map $\Theta - i\ad \mu$ acts as
\begin{align*}
(\Theta - i\ad \mu)(X_{i\alpha}) &= \left(d\mu + \left[\mu, \frac{e^{-i \ad Y} -1}{\ad Y}\right] \right) \frac{\ad Y}{\sin \ad Y} X_{i\alpha}\\ &= \left(d\mu + \left[\mu, \frac{e^{-i \ad Y} -1}{\ad Y}\right] \right) \frac{\alpha(Y)}{\sinh \alpha(Y)} X_{i\alpha}.
\end{align*}
Now, by $\Ad G$-invariance of $\mu$,
\begin{align*}
 (d\mu)_Y([X,Y]) = \frac{d}{dt}|_{t=0} \mu\left( e^{tX}Ye^{-tX} \right) = \frac{d}{dt}|_{t=0} \Ad_{e^{tX}} \mu(Y ) = [X,\mu(Y)],
\end{align*}
for all $Y\in \mathfrak{t}$ and $X \in \mathfrak{g}$. Therefore, $$\alpha(Y) d\mu ( X_{i\alpha}) = \alpha(\mu(Y)) X_{i \alpha},$$ and so
\begin{align*}
\left( d\mu  +  \left[\mu, \frac{e^{-i \ad Y} -1}{\ad Y}\right] \right)  \frac{\alpha(Y)}{\sinh \alpha(Y)} X_{i\alpha} = \frac{\alpha(\mu(Y))}{\alpha(Y)} \cdot \alpha(Y) \left( \frac{\cosh\alpha(Y)}{\sinh\alpha(Y)} + 1    \right)   X_{i \alpha}.
\end{align*}
Thus  the eigenvalues of $\Theta - i \ad \mu$ are given by the eigenvalues of the Hessian of $\widetilde{f}: \mathfrak{t} \rightarrow \mathfrak{t}$ as well as by (\ref{almuY}).
\end{proof}
\begin{proof}[Proof of  Proposition \ref{prp:con_psh}]
Let $\widetilde{f}$ be a $W(G,\textbf{T})$-invariant convex function on $\mathfrak{t}$. Then $\alpha(\nabla \widetilde{f}) / \alpha(Y) \geq 0$ for all $Y \in \mathfrak{t}$.\footnote{To see this, let $w_\alpha \in W(G,\textbf{T})$ be the reflection associated to the root $\alpha$. Pick a basis $\{e_1 \dots, e_k\}$ of $\mathfrak{t}$ such that $e_2, \dots e_k$ are in $\ker \alpha$ and $w_\alpha(e_1) = -e_1$. Write $(y_i)$ for the coordinates of $Y \in \mathfrak{t}$ with respect to this basis. Then $\widetilde{f}(y_1, \dots,  y_n) = \widetilde{f}(-y_1, \dots , y_n)$  by Weyl-invariance of $\widetilde{f}$. Also taking into account the convexity of $\widetilde{f}$,  we see that $\alpha( \nabla f) / \alpha(Y)= \frac{1}{y_1} \frac{\partial \widetilde{f}}{\partial y_1}$ is non-negative.} Because $x (\frac{\cosh x}{\sinh x} +1) \geq 0$ for all $x \in \mathbb{R}$, the expression $\alpha(\mu(Y)) \left( \frac{\cosh \alpha(Y)}{\sinh \alpha(Y)} + 1  \right)$ is therefore non-negative for each $\alpha \in R$. Moreover, convexity of $\widetilde{f}$ also implies that the eigenvalues of its Hessian matrix are non-negative. Therefore, Lemma \ref{lem:psh_theta} and  Lemma \ref{lem:theta_eigenvalues} imply that the $G \times G$-invariant extension of $\widetilde{f}$ to a function $f$ on $G^\mathbb{C}$ is plurisubharmonic. The converse follows from the fact that the eigenvalues of the Hessian of $\widetilde{f}$ are also eigenvalues of $\Theta - i \ad \mu$.
\end{proof}

The following theorem is now obtained as a corollary of  Proposition \ref{prp:con_psh}.
\begin{thm}
\label{thm:can_semineg}
 The canonical line bundle on $T^*G$ is semi-negative. 
\end{thm}
\begin{proof}
  By  Proposition \ref{prp:con_psh}, it is sufficient to prove that the function
\begin{align*}
 Y \mapsto \log(\eta(Y)) = \sum_{\alpha} \log\left( \frac{\sinh \alpha(Y)}{\alpha(Y)} \right)   
\end{align*}
is convex on $\mathfrak{t}$. Since a finite sum of convex functions is convex, we need that 
\begin{align*}
 Y \mapsto \log\left( \frac{\sinh \alpha(Y)}{\alpha(Y)} \right) = \log\left( \frac{\sinh |\alpha(Y)|}{|\alpha(Y)|}\right)
\end{align*}
is convex on $\mathfrak{t}$  for each real root $\alpha$.  
Now, for each $\alpha \in R$, the map $Y \mapsto |\alpha(Y)|$ is a semi-norm on $\mathfrak{t}$, so in particular,  each of these maps is convex.  By applying the chain rule it is not difficult to show that, if $\phi: \mathbb{R} \rightarrow \mathbb{R}$ is convex and non-decreasing on the range of another convex map $f: \Omega \rightarrow \mathbb{R}$, where $\Omega$ is a convex set, then  the composition $\phi \circ f$ is convex,  too. So, we need to prove that the function $\log \circ \widetilde{\eta}: \mathbb{R} \rightarrow \mathbb{R}$ is convex and non-decreasing on $\mathbb{R}^{\geq 0}$, where $\widetilde{\eta}: \mathbb{R} \rightarrow \mathbb{R}$ is given by
\begin{align*}
t \mapsto \frac{\sinh t}{t}.
\end{align*}
By looking at the power series for $\widetilde{\eta}$ one sees that $\log \circ \widetilde{\eta}$ is non-decreasing on $\mathbb{R}^{\geq 0}$.  The function $\log \circ \widetilde{\eta}$ is convex if and only if
\begin{align*}
 \widetilde{\eta}\widetilde{\eta}''(t) -\widetilde{\eta}'(t) ^2 \geq 0 \quad \text{for all }t \in \mathbb{R}.
\end{align*}
The first and second derivatives of $\widetilde{\eta}$ are equal to
\begin{align*}
\widetilde{\eta} '(t) = \frac{\cosh t}{t} - \frac{\sinh t}{t^2}, \quad \widetilde{\eta}''(t) = \frac{\sinh t}{t} - \frac{2 \cosh t}{t^2} + \frac{2\sinh t}{t^3}. 
\end{align*}
The expression $\widetilde{\eta}(t) \widetilde{\eta}''(t) - \widetilde{\eta}'(t) ^2 $ is then equal to
\begin{align*}
\widetilde{\eta}(t) \widetilde{\eta}''(t) - \widetilde{\eta}'(t) ^2  &= \frac{\sinh^2 t}{t^2} - \frac{2 \cosh t \sinh t}{t^3} + \frac{2 \sinh^2 t}{t^4}\\ &- \frac{\cosh^2 t}{t^2} + \frac{2\cosh t \sinh t}{t^3} - \frac{\sinh^2 t}{t^4} \\
&= \frac{\sinh^2 t - \cosh^2 t}{t^2} + \frac{\sinh^2 t}{t^4} = \frac{\sinh^2 t - t^2}{t^4} > 0,
\end{align*}
for all $t \in \mathbb{R} \setminus \{0 \}$. Since $\widetilde{\eta}$ is smooth, $\widetilde{\eta}(t) \widetilde{\eta}''(t) - \widetilde{\eta}'(t) ^2 \geq 0$ for all $t \in \mathbb{R}$. Hence $\log \widetilde{\eta}$ is convex and the theorem is proved.
\end{proof}

\subsection{Dolbeault--Dirac quantization}
\label{sct:DD_spin_quantization}
In this section we prove that the Dolbeault--Dirac quantization  of $T^*G$ is $G \times G$-equivariantly isomorphic to the Hilbert space \eqref{eq:Hallq}, which was previously obtained by Hall \cite{Hal02} via geometric quantization. To achieve this we use Theorem \ref{thm:kvt} together with the semi-negativity of the canonical line bundle on $T^*G$ (i.e., Theorem \ref{thm:can_semineg}) to show that the kernel of  the Dolbeault--Dirac operator only consists of holomorphic \emph{functions} on $T^*G$ that are square-integrable with respect to some $G\times G$-invariant measure.

Define a smooth function $\phi: G \times \mathfrak{g} \rightarrow \mathbb{R}$ by 
\begin{align*}
 \phi(x,Y) = \pi |Y|^2.
\end{align*}
Recall that the pre-quantum line bundle $L$ on $T^*G$ is simply the trivial hermitian line bundle endowed with the connection $\nabla^L = d + 2\pi i \theta$, where $\theta$ denotes the canonical symplectic potential.  We take the $G \times G$-action on the fiber direction of $L$ to be trivial. Since $\theta$ is a $G \times G$-invariant $1$-form, the connection $\nabla^L$ is $G\times G$-invariant, too.

\begin{prp}
\label{prp:holoL}
Let $\alpha \in \Gamma^\infty(T^*G,\Lambda^\bullet(T^{*(0,1)}M) \otimes L)$ be a smooth section. Then $\ol{\partial}^L \alpha =0$ if and only if $\ol{\partial} ( e^\phi \alpha) = 0$, with $\phi$ as above.
\end{prp}
\begin{proof}
 We show that $\ol{\partial} \phi = 2\pi i (\pi^{(0,1)} \theta)$. This implies that
\begin{align*}
 \ol{\partial}(e^\phi \alpha ) =e^\phi  \left ( \ol{\partial}\alpha + (\ol{\partial} \phi) \wedge \alpha \right) = e^\phi\left( \ol{\partial} +  2 \pi i \pi^{(0,1)} \theta \right) \alpha = e^\phi \ol{\partial}^L \alpha,
\end{align*}
so that $\ol{\partial}^L \alpha =0$ if and only if $\ol{\partial} (e^\phi \alpha) = 0$.

With respect to the basis $\{\alpha_k,dy_k\}$ (see the remarks above  Theorem \ref{thm:TGgeodesicallycomplete} for the definition of these $1$-forms) the form $d\phi$ is equal to $2\pi y_kdy^k$, which we simply write as $(0,2\pi Y)$. Since $\ad Y$ acts trivially on $Y$, we compute
\begin{align*}
 \ol{\partial} \phi  = (0,\pi Y) + (i \pi Y,0) = (i\pi Y, \pi Y). 
\end{align*}

On the other hand, by Lemma \ref{lem:formulas}, 
\begin{align*}
\theta_{(g,Y)}(Z_1,Z_2) = \langle Y, Z_1 \rangle_\mathfrak{g} = \sum_{k=1}^n y_k z_k,
\end{align*}
when $Z_1 = \sum_{k=1}^n z_k e_k$. Consequently, $\theta_{(g,Y)} = (Y,0)$ in the basis $\{\alpha_k,dy^k\}$. Again,
\begin{align*}
\pi^{(0,1)} \theta_{(g,Y)}  = \frac{1}{2}\left( \theta_{(g,Y)} - iJ \theta_{(g,Y)} \right) = \frac{1}{2}\left( (Y,0) -  (0,iY) \right) = \frac{1}{2}( Y, -i Y).  
\end{align*}
Multiplying by $2\pi i$ gives $2\pi i \pi^{(0,1)}\theta = (i \pi Y, \pi Y)$, which is equal to $\ol{\partial}\phi$.
\end{proof}

If the bundle $L$ is endowed with the unique holomorphic structure for which the section $e^{-\phi}$ is holomorphic, then $\nabla^L = d +2\pi i \theta$ is the unique Chern connection on $L$. As $(\nabla^L)^2 = 2 \pi i \omega$, the line bundle $L$ is positive.

Since $D^L$ is elliptic, the kernel of $\ol{D}^L=(D^L)^*$ consists of smooth sections. So, to determine the kernel of $(D^L)^*$ it is important to know how $(D^L)^*$ acts on smooth sections. The following lemma shows that for any symmetric differential operator $D$, the action of $D^*$ on smooth sections coincides with the action of $D$ as a differential operator.
\begin{lem}
\label{lem:disdif} 
Let $D$ be a symmetric differential operator on a hermitian vector bundle $E$ over an oriented Riemannian manifold $M$ (with domain $\Gamma^\infty_c(M,E)$). If $s \in \Gamma^\infty(M,E)\cap \mathcal{D}(D^*)$, then 
\begin{align*}
D^*s = Ds.
\end{align*}
\end{lem}
\begin{proof}
 Suppose that $s \in \Gamma^\infty(M,E)\cap \mathcal{D}(D^*)$. Let $t \in \Gamma^\infty_c(M,E)$ be given and let $\psi \in C^\infty_c(M)$ be such that $0 \leq \psi \leq 1$ and $\psi \equiv 1$ in a neighbourhood $U$ of $\text{supp }t$. Then $\langle s, Dt \rangle = \langle \psi s, Dt \rangle$, since $D$ is a local operator. The section $\psi s$ is compactly supported and smooth, so that by symmetry of $D$:
\begin{align*}
\langle s, Dt \rangle &= \langle \psi s, Dt \rangle = \langle D(\psi s), t \rangle = \int_M \langle D(\psi s)(x), t(x) \rangle_x dx,
\end{align*}
where $D$ acts on $\psi s$ as a differential operator. Now, $\langle D(\psi s)(x), t(x) \rangle_x$ is zero outside $\text{supp }t$ and the differential operator $D$ commutes with $\psi$ on the neighbourhood $U$,  since $\psi \equiv1$ there. Consequently, 
\begin{align*}
\langle D(\psi s)(x), t(x) \rangle_x = \langle \psi(x) Ds(x),t(x) \rangle_x,
\end{align*}
for all $x \in M$, and so
\begin{align*}
\langle s, Dt \rangle &=\int_M \langle D(\psi s)(x), t(x) \rangle_x dx =  \int_M \langle \psi(x) (Ds)(x), t(x) \rangle_x dx =  \langle \psi Ds,t \rangle \\ &= \langle Ds,t \rangle,
\end{align*}
where $D$ acts as a differential operator on $s$. This  holds for all $t \in \Gamma^\infty_c(M,E)$.
\end{proof}

We are now ready to determine the Dolbeault--Dirac quantization of $T^*G$. Let $\mathcal{H}L^2(T^*G, e^{-2\pi|Y|^2} \varepsilon)$  be the Hilbert space of all holomorphic functions on $T^*G$ that are square-integrable with respect to the $G\times G$-invariant measure $e^{-2\pi|Y|^2} \varepsilon$.
\begin{thm}
\label{thm:quantDD_cot}
The Dolbeault--Dirac quantization $\ker \ol{D}^L $ of $T^*G$ (with its standard K\"ahler structure) 
 is $G \times G$-equivariantly isomorphic to $\mathcal{H}L^2(T^*G, e^{-2\pi|Y|^2} \varepsilon)$.
\end{thm}
\begin{proof}
By Theorem \ref{thm:can_semineg} the canonical line bundle $K$ is semi-negative, so that $K^*$ is semi-positive. Since the total curvature of a tensor product of two line bundles is the sum of the curvatures of the individual line bundles, we see that $K^* \otimes L$ is positive. According to Theorem \ref{thm:kvt} the kernel of $\ol{D}^L$ is contained in $\Gamma^\infty(T^*G,L)$. Since $T^*G$ is geodesically complete, we have $\ol{D}^L=(D^L)^*$, hence by Lemma \ref{lem:disdif},
\begin{align*}
 \ker \ol{D}^L = \ker (D^L)^* = \{ s \in \Gamma^\infty(T^*G,L) \cap L^2(T^*G,L)\mid  D^L s =0 \}.
\end{align*}
Now, $D^L s =0$ for a smooth section $s$ of $L$, if and only if $\ol{\partial}^L s = 0$ if and only if $s$ is a holomorphic section of $L$. By Proposition \ref{prp:holoL} the holomorphic sections of $L$ are of the form $s = fe^{-\phi}$, with $\phi = \pi|Y|^2$ and $f$ a holomorphic function on $T^*G$. Thus
\begin{align*}
 \ker \ol{D}^L &= \left\{ fe^{-\phi} \biggm| f \text{ holomorphic and }\int_M |f|^2 e^{-2 \pi|Y|^2} \varepsilon < \infty \right\} \\ &\cong \mathcal{H}L^2(T^*G,e^{-2\pi|Y|^2} \varepsilon).
\end{align*}
Since $e^{-\phi}$ is $G \times G$-invariant, the last isomorphism intertwines the $G \times G$-actions on both spaces.
\end{proof}

\noindent Combining Theorem \ref{thm:quantDD_cot}  with   \cite[Theorem 10]{Hal94} and
Proposition \ref{prp:Hall_iso}, we obtain:
\begin{thm}
\label{thm:quantspin}
 Let $G$ be a compact connected Lie group and endow $T^*G$ with its standard K\"ahler structure. Let 
 $(T^*G \times \mathbb{C}, d + 2\pi i \theta)$  be the asociated 
 $G \times G$-equivariant pre-quantization line bundle.
  Then the Dolbeault--Dirac quantization  $ \ker \ol{D}^L$ of $T^*G$ is 
   $G \times G$-equivariantly 
  equal to Hall's Hilbert space $ \mathcal{H}L^2(T^*G, e^{-2\pi|Y|^2} \varepsilon)$,
  equipped with the natural $G \times G$-action. 
Consequently  (\textit{cf.} \cite[Theorem 2.6]{Hal02}), our Dolbeault--Dirac quantization of $T^*G$ is $G \times G$-equivariantly unitarily isomorphic to $L^2(G)$.
\end{thm}
\section{Quantization of the symplectic stratified quotient}
\label{sct:quantization_cot_max_torus}
Let $G$ be a compact connected Lie group, and
 consider the action of $G$ on $T^*G$ induced by the action of $G$ on itself by conjugation. 
Since this induced action  is not free (like the one on $G$), the ensuing Marsden--Weinstein quotient  $T^*G /\!/ \Ad G$ is not a manifold, but merely a symplectic stratified space \cite{SL91}. 
 In this section we first explicitly determine the stratified structure of  $T^*G /\!/ \Ad G$, and then define the quantization of the Marsden--Weinstein quotient to be the Dolbeault--Dirac quantization of its principal stratum.
 On this definition, quantization turns out to commute with (singular) reduction, which provides an \emph{a posteriori} justification of our definition. 
For a  different, \emph{a priori} justification of this definition, we note that because our singular quotient  $T^*G /\!/ \Ad G$
is stratified by symplectic spaces, its singular (i.e., non-principal) strata have codimension at least $2$. 
This makes the following two results (though restricted to flat space) rather suggestive.
\begin{prp}
\label{prp:sobolevgeneral}
 The space $C_c^\infty(\mathbb{R}^n \setminus \mathbb{R}^{n-k})$ is dense in $H^s(\mathbb{R}^n)$ if and only if $k>0$ and  $s \leq \frac{k}{2}$ (where, for $n=k$, $\mathbb{R}^0 =\{0\} \subset \mathbb{R}^n$).
\end{prp}
Cf.\ \cite[Theorem 3.1.3]{Boe14} for the proof.
Now endow $\mathbb{R}^{2n}$ with its standard K\"ahler structure. The bundle $\Lambda^{(0,\bullet)}T^*\mathbb{R}^{2n}$ is identified with the direct sum of $2^n$-copies of $L^2(\mathbb{R}^{2n})$ by considering the unitary frame 
\begin{align*}
\left\{ \tfrac{1}{\sqrt{2^k}} d\ol{z}^{i_1} \wedge \cdots \wedge d\ol{z}^{i_k}   \biggm|  1 \leq i_1 < i_2 < \cdots  <i_k, 0 \leq k \leq n \right\}.
\end{align*}
The Dolbeault--Dirac operator has constant coefficients  with respect to this unitary frame,  so by \cite[Theorem 6.24]{Gru08}, the graph norm of the Dolbeault--Dirac operator on $\Gamma^\infty_c\left( \mathbb{R}^{2n},\Lambda^{(0,\bullet)}T^*\mathbb{R}^{2n}\right)$ is equivalent to the $H^1$-norm. Combining this result with Proposition \ref{prp:sobolevgeneral} we obtain the following:
\begin{cor}
\label{cor:denseDD}
Let $2 \leq k \leq 2n$ be a natural number. The subspace 
\begin{align*}
\Gamma^\infty_c(\mathbb{R}^{2n} \setminus \mathbb{R}^{2n-k},\Lambda^{(0,\bullet)}T^*\mathbb{R}^{2n})
\end{align*}
 is dense in $\Gamma^\infty_c(\mathbb{R}^{2n}, \Lambda^{(0,\bullet)}T^*\mathbb{R}^{2n})$ with respect to the graph norm of  the Dolbeault--Dirac operator on $\mathbb{R}^{2n}$.
\end{cor}
Consequently, at least on flat space the Dolbeault--Dirac operator is insensitive to the removal of submanifolds of codimension at least $2$. In our context, this suggests that it may be restricted to the principal stratum of $T^*G /\!/ \Ad G$ without any loss. 
\subsection{Geometry of the reduced space}
Given our  compact connected Lie group $G$,  we fix  a maximal torus  $\textbf{T}$ in $G$ with Lie algebra $\mathfrak{t}$
and  corresponding Weyl group $W(G,\textbf{T})$. As in \S\ref{sct:cot_bun}, we choose an $\Ad G$-invariant inner product on $\mathfrak{g}$ and use this inner product and left-trivialisation of $T^*G$ to identify $T^*G$ with $G \times \mathfrak{g}$.

 Calculations similar to those in \cite[Section 4.4]{AM78} lead to the following result.
\begin{lem}
\label{lem:moment_map}
 The momentum map $j: G \times \mathfrak{g} \rightarrow \mathfrak{g}^*$, considered as a map from $G \times \mathfrak{g}$ into $\mathfrak{g}$, is equal to
\begin{align*}
 j(g,Y) = \Ad_gY - Y \in \mathfrak{g},
\end{align*}
where $g \in G$, $Y \in \mathfrak{g}$.
\end{lem}
Having an explicit formula for the momentum map, we now determine the strata of the Marsden--Weinstein quotient 
(\ref{defMWQ}). We show that it is sufficient to know how these strata intersect $\textbf{T} \times \mathfrak{t}$. The isotropy group of an element $g \in G$ under the action of $G$ on itself by conjugation is simply the \emph{centraliser} $\mathcal{Z}_G(g)$. Similarly, we write $\mathcal{Z}_G(Y)$ for the isotropy group of $Y \in \mathfrak{g}$ under the adjoint action of $G$ on $\mathfrak{g}$, and refer to $\mathcal{Z}_G(Y)$ as the \emph{centraliser} of $Y$, too.
\begin{lem}
\label{lem:mwred}
 Let $(x,Y) \in j^{-1}(0) \subset G \times \mathfrak{g}$ be arbitrary. The orbit of $(x,Y)$ under the adjoint action of $G$ contains an element of $\textbf{T} \times \mathfrak{t}$, and 
 we have a homeomorphism 
 $$
 T^*G /\!/ \Ad G \cong (\textbf{T} \times \mathfrak{t}) / W(G,\textbf{T}) ,
 $$ 
 with respect to the obvious diagonal action of $W(G,\textbf{T})$ on $\textbf{T} \times \mathfrak{t}$.
\end{lem}
This result is not new \cite{HRS09}, but nonetheless we provide an independent proof, since it
 contains manipulations that will be needed in what follows.
\begin{proof}
 Let $(x,Y) \in j^{-1}(0)$ and consider the one-dimensional Lie subalgebra $\mathfrak{h}_Y \subset \mathfrak{g}$ generated by $Y$. The Lie group $\exp(\mathfrak{h}_Y)$ is a torus in $\mathfrak{g}$. Since $j(g,Y)=0$, we have
\begin{align*}
 g\exp(tY)g^{-1} = \exp(\Ad_g tY) = \exp(tY), \quad (t \in \mathbb{R}),
\end{align*}
so that $g$ centralises the torus $\exp(\mathfrak{h}_Y)$. Consequently, there is a maximal torus in $G$ that contains both $\exp(\mathfrak{h}_Y)$ and $g$ (see \cite[Theorem 4.50]{Kna96}). Since any two maximal tori are conjugate in $G$, there exists an $h \in G$ such that $\textbf{T}$ contains both $hgh^{-1}$ and $\exp(\Ad_h\mathfrak{h}_Y)$. In particular, the path $t \mapsto \exp(t \Ad_h Y)$ lies in $\textbf{T}$, so that $\Ad_h Y \in \mathfrak{t}$. Thus  $(hgh^{-1}, \Ad_h Y)$ lies in $\textbf{T} \times \mathfrak{t}$.

Clearly, $\textbf{T} \times \mathfrak{t} \subset j^{-1}(0)$ and this inclusion induces a continuous map 
\begin{align*}
\iota: (\textbf{T} \times \mathfrak{t}) / W(G,\textbf{T}) \rightarrow j^{-1}(0) / \Ad G.
\end{align*}
 By the previous paragraph, $\iota$ is surjective. To show that it is injective, we proceed in the same way as in \cite[Proposition 4.53]{Kna96}, where it is shown that $G/\Ad G$ is homeomorphic to $\textbf{T}/W(G,\textbf{T})$. Suppose that $(g,Y), (g',Y') \in \textbf{T} \times \mathfrak{t}$ are in the same $G$-orbit. That is, there exists $h\in G$ such that $(g',Y') = (hgh^{-1},\Ad_h Y)$. Consider the closed Lie subgroup $\mathcal{Z}_G(g,Y) =\mathcal{Z}_G(g) \cap \mathcal{Z}_G(Y)$, i.e., the closed Lie subgroup of $G$ consisting of all elements that centralise both $g$ and $Y$. Its Lie algebra is  
\begin{align*}
\mathcal{Z}_\mathfrak{g}(g,Y):= \{ X \in \mathfrak{g} \mid \Ad_g X = X \text {and } [X,Y]=0 \},
\end{align*}
and $\mathcal{Z}_\mathfrak{g}(g,Y)$ contains $\mathfrak{t}$. It also contains $\Ad_h^{-1} \mathfrak{t}$: pick $X \in \mathfrak{t}$, then
\begin{align*}
 \Ad_g \Ad_{h^{-1}} X = \Ad_{h^{-1}} \Ad_{hgh^{-1}} X = \Ad_{h^{-1}}X,
\end{align*}
where we used the fact that $hg^{-1}h \in \textbf{T}$. Similarly,
\begin{align*}
 [\Ad_{h^{-1}} X,Y] = \Ad_{h^{-1}} [ X, \Ad_h Y] = 0.
\end{align*}
Both $\mathfrak{t}$ and $\Ad_{h^{-1}} \mathfrak{t}$ are maximal abelian subalgebras in $\mathcal{Z}_\mathfrak{g}(g,Y)$. Hence
 there exists an element $k$ in (the identity component of) $\mathcal{Z}_G(g,Y)$ such that $\Ad(kh^{-1}) \mathfrak{t} = \mathfrak{t}$. Consequently, $kh^{-1} \in N_G(\textbf{T})$, and 
\begin{align*}
kh^{-1} \cdot (hgh^{-1}, \Ad_h Y) = k \cdot (g,Y) = (g,Y),
\end{align*}
so that $(g,Y)$ and $(g',Y')$ are in the same $W(G,\textbf{T})$ orbit. Thus  $\iota$ is a continuous bijection. 
To prove that $\iota$ is a homeomorphism we show that it is closed. For this, it is sufficient that the map $\textbf{T} \times \mathfrak{t} \rightarrow j^{-1}(0) /\Ad G$ is closed. Since $\textbf{T} \times \mathfrak{t}$ is a closed subset of $j^{-1}(0)$, it is in turn sufficient to show that the projection $j^{-1}(0) \rightarrow j^{-1}(0)/\Ad G$ is closed, which is what we are going to do now. By compactness of $G$, the map 
\begin{align*}
\Phi: G \times j^{-1}(0) \rightarrow j^{-1}(0), \quad (g,x) \mapsto gx
\end{align*} 
is proper. Every proper map into a locally compact Hausdorff space is closed, and so in particular $\Phi$ is closed. Therefore, if $C$ is a closed set of $j^{-1}(0)$, then $GC = \Phi(G \times C)$ is closed. Thus  the quotient map $j^{-1}(0) \rightarrow j^{-1}(0)/\Ad G$ is closed. Consequently, the continuous bijection $$\iota:(\textbf{T} \times \mathfrak{t})/W(G,\textbf{T}) \rightarrow j^{-1}(0) / \Ad G$$ is closed, and hence is a homeomorphism.
 \end{proof}

The restriction of the symplectic structure on $G \times \mathfrak{g}$ to $\textbf{T} \times \mathfrak{t}$ is equal to
\begin{align*}
  \omega_{(g,Y)}((X_1,X_2),(Z_1,Z_2) ) = \langle X_2, Z_1 \rangle_\mathfrak{g} - \langle X_1, Z_2 \rangle,
\end{align*}
where $g \in \textbf{T}, Y, X_{1,2}, Z_{1,2} \in \mathfrak{t}$ (see \eqref{eq:symp}). Moreover, on $\textbf{T} \times \mathfrak{t}$, the map \eqref{eq:differentialPhi} is just the identity when restricted to $\mathfrak{t} \times \mathfrak{t}$. Therefore, the induced complex structure on $\textbf{T} \times \mathfrak{t}$, as a map from $\mathfrak{t} \times \mathfrak{t}$ onto itself after identifying each tangent space with $\mathfrak{t} \times \mathfrak{t}$ using left-trivialisation, is simply
\begin{align*}
 J(X_1,X_2) = (-X_2,X_1), \quad (X_1,X_2 \in \mathfrak{t}).
\end{align*}  
Obviously, the K\"ahler structure that $\textbf{T} \times \mathfrak{t}$ inherits from $G \times \mathfrak{g}$ is equal to the standard K\"ahler structure on $\textbf{T} \times \mathfrak{t}$. One verifies directly that this K\"ahler structure on $\textbf{T} \times \mathfrak{t}$ is invariant under the action of the Weyl group.

For a strongly Hamiltonian, proper group action on a symplectic manifold $M$, the space $j^{-1}(0) / \Ad G$ is stratified by the strata \cite{OR04,SL91} 
\begin{eqnarray}
 M^{(H)}_{x,0}&=& (j^{-1}(0) \cap G M_H^x) / \Ad G;\nonumber \\
  &=&  (j^{-1}(0) \cap M_H^x) / (N_G(H)^x/ H), \label{eq:srs}
\end{eqnarray}
where $H \subset G$ is an isotropy group for the $G$-action on $M$,  $M_H^x$ is the connected component of 
$M_H := \{ x \in M \mid H \text{ is the isotropy group of }x \}$
that contains $x$ (assuming that $x \in M_H$), and 
\begin{align*}
 N_G(H)^x := \{ g \in G \mid g \cdot M_H^x = M_H^x \}. 
\end{align*} 

We now use (\ref{eq:srs}) to determine the strata of $j^{-1}(0)/\Ad G$ for the adjoint action of $G$ on $T^*G$. By Lemma \ref{lem:mwred} it is sufficient to consider only those submanifolds $(G \times \mathfrak{g})_H^{(g,Y)}$ that have non-empty intersection with $\textbf{T} \times \mathfrak{t}$. But if $(g,Y)$ is in $(\textbf{T} \times \mathfrak{t}) \cap (G \times \mathfrak{g})_H^{(g,Y)}$, then $H \supset \textbf{T}$. Conversely, if $H \supset \textbf{T}$, then necessarily $(G \times \mathfrak{g})_H \subset \textbf{T} \times \mathfrak{t}$, because $(G \times \mathfrak{g})^\textbf{T} = \textbf{T} \times \mathfrak{t}$ . This latter fact can be seen as follows: if $(g,Y) \in (G\times \mathfrak{g})^\textbf{T}$, then $\textbf{T}$ centralises both $g$ and $Y$. In particular, there exists a maximal torus containing both $g$ and $\textbf{T}$ and so $g \in \textbf{T}$, because $\textbf{T}$ is a maximal torus. Similarly, $Y \in \mathfrak{t}$. Thus  we have proved:

\begin{prp}
\label{prp:restricttotorus}
\begin{enumerate}
\item If $H \supset \textbf{T}$, then 
\begin{align*}
j^{-1}(0) \cap (G \times \mathfrak{g})^{(g,Y)}_H = (G \times \mathfrak{g})^{(g,Y)}_H  \subset \textbf{T} \times \mathfrak{t}. 
\end{align*}
\item Each stratum of $j^{-1}(0)/\Ad G$ is of the form 
\begin{align*}
(G \times \mathfrak{g})^{(g,Y)}_H / (N_G(H)^{(g,Y)}/H),
\end{align*}
with $H \supset \textbf{T}$. 
\end{enumerate}
\end{prp}
Proposition \ref{prp:restricttotorus} basically says that the symplectic stratification on $j^{-1}(0)/\Ad G$ is obtained by partitioning $\textbf{T} \times \mathfrak{t}$ into the connected components of $(G \times \mathfrak{g})_H$ for $H \supset \textbf{T}$ and subsequently project these into $j^{-1}(0)/\Ad G$.

 We now further analyze this partition of $\textbf{T} \times \mathfrak{t}$. 
Consider the action of $G$ on itself by conjugation and let $\textbf{T}$ be a maximal torus. The principal stratum of $G$ is   
\begin{align*}
G_{princ} =\{ g \in G \mid \mathcal{Z}_G(g)  \text{ is a maximal torus}\}.
\end{align*}
Similarly, the principal stratum of the adjoint action of $G$ on $\mathfrak{g}$ is equal to
\begin{align*}
\mathfrak{g}_{princ} =  \{ X \in \mathfrak{g} \mid \mathcal{Z}_G(X)  \text{ is a maximal torus}\}.
\end{align*}
Proofs of these facts may be found in \cite[Theorem 3.7.1 and Corollary 3.3.2]{DK00}. In particular, $G_{princ}$ and $\mathfrak{g}_{princ}$ are open and dense in $G$ and $\mathfrak{g}$, respectively. Since $G / \Ad G \cong \textbf{T} /W$ and the projection $G \rightarrow G /\Ad G$ is open, the set $G_{princ} \cap \textbf{T}  = G_\textbf{T}$ is open and dense in $\textbf{T}$. Similarly, $\mathfrak{g}_\textbf{T}$ is open and dense in $\mathfrak{t}$. 

The principal stratum of $G \times \mathfrak{g}$ is $(G\times \mathfrak{g})_{C(G)}$, where $C(G)$ denotes the centre of $G$, but this stratum does not intersect $j^{-1}(0)$, because each element in $j^{-1}(0)$ is at least fixed by some maximal torus of $G$. Instead, the principal stratum on $\textbf{T} \times \mathfrak{t}$ is as described in the following proposition.
\begin{prp}\label{prp:ODSP}
 The space $(G \times \mathfrak{g})_\textbf{T} \subset \textbf{T} \times \mathfrak{t}$ is an open and dense submanifold of $\textbf{T} \times \mathfrak{t}$.
\end{prp}
\begin{proof}
 The submanifold $(G \times \mathfrak{g})_\textbf{T}$ is non-empty, since both $G_\textbf{T}$ and $\mathfrak{g}_\textbf{T}$ are non-empty. Let $(g,Y) \in (G \times \mathfrak{g})_\textbf{T}$ be given. The dimension of $(G \times \mathfrak{g})^{(g,Y)}_\textbf{T}$ is equal to the dimension of $(G \times \mathfrak{g})_{(g,Y)}^\textbf{T} = \textbf{T} \times \mathfrak{t}$. In particular, $(G \times \mathfrak{g})_\textbf{T}$ is an open submanifold of $\textbf{T} \times \mathfrak{t}$. Here, $(G \times \mathfrak{t})^\textbf{T}_{(g,Y)}$ denotes the connected component of $(G \times \mathfrak{g})^\textbf{T}$ containing $(g,Y)$. To see that $(G \times \mathfrak{g})_\textbf{T}$ is dense in $\textbf{T} \times \mathfrak{t}$, note that $G_\textbf{T} \times \mathfrak{g}_\textbf{T}$ is contained in $(G \times \mathfrak{g})_\textbf{T}$ and that $G_\textbf{T} \times \mathfrak{g}_\textbf{T}$ is dense in $\textbf{T} \times \mathfrak{t}$ by the observations preceding this proposition.
\end{proof}
\subsection{Quantization after reduction}
\label{sct:rq}
We are now ready to define the Dolbeault--Dirac quantization of the singular quotient $j^{-1}(0)/\Ad G$.
\begin{dfn}\label{def:QMWQ}
 The quantization of the Marsden Weinstein quotient 
\begin{align*}
T^*G/\!/ \Ad G = j^{-1}(0)/ \Ad G
\end{align*}
is the Dolbeault--Dirac quantization of the principal stratum of the quotient.
\end{dfn}
By Proposition \ref{prp:ODSP}, this  principal stratum is the
manifold $(G \times \mathfrak{g})_\textbf{T} / W(G,\textbf{T})$, seen as an open and dense subset of $j^{-1}(0) / \Ad G$.

Let $P$ denote the complement of $(\textbf{G} \times \mathfrak{g})_\textbf{T}$ in $\textbf{T} \times \mathfrak{t}$. By \cite[Corollary 1, pp. 316]{Bou05} there is only a finite number of closed subgroups of $G$ containing a given maximal torus $\textbf{T} \subset G$. Therefore, the union 
\begin{align*}
 P = \cup_{H \supsetneq \textbf{T}, \text{isotropy group}} (G \times \mathfrak{g})^H
\end{align*}
is finite. The connected components of $(G \times \mathfrak{g})^H$ are closed submanifolds of $\textbf{T} \times \mathfrak{t}$. If one of these connected components were open as well, then it would be a non-empty, closed and open subset of the connected space $\textbf{T} \times \mathfrak{t}$ and therefore would be equal to $\textbf{T} \times \mathfrak{t}$. This is impossible if $H \supsetneq \textbf{T}$, because $(G \times \mathfrak{g})_\textbf{T}$ is non-empty. Consequently, the connected components of $(G \times \mathfrak{g})^H$ are of lower dimension than $\textbf{T} \times \mathfrak{t}$ if $H \supsetneq \textbf{T}$. Since these connected components are symplectic manifolds, their codimensions as submanifolds of the symplectic manifold $\textbf{T} \times \mathfrak{t}$ are at least $2$. 

Denote the dimension of $\textbf{T}$ by $n$ and write $\mathbb{E}$ for the bundle $\Lambda^{(0,\bullet)}T^*M$.
\begin{prp}
\label{prp:denseMH}
 Let $H \supsetneq \textbf{T}$ be an isotropy group for the adjoint action of $G$ on $G \times \mathfrak{g}$.  Then $\Gamma_c^\infty(\textbf{T} \times \mathfrak{t} \setminus (G \times \mathfrak{g})^H,\mathbb{E})$ is dense in $\Gamma_c^\infty(\textbf{T} \times \mathfrak{t},\mathbb{E})$ with respect to the graph norm of $D$, the Dolbeault--Dirac operator on $\textbf{T} \times \mathfrak{t}$.
\end{prp}
\begin{proof}
Choose $\epsilon >0$ such that the exponential map $\exp: \mathfrak{g} \rightarrow G$ is a diffeomorphism from the open $G$-invariant neighbourhood $U=\{Y \in \mathfrak{g} \mid |Y|<\epsilon \}$ onto an open neighbourhood $V \subset G$ of $e_G$. Let $H \supsetneq \textbf{T}$ be an isotropy group for the action of $G$ on $G \times \mathfrak{g}$. For each $$(g,Y) \in (G \times \mathfrak{g})^H = G^H \times \mathfrak{g}^H,$$ we consider the chart 
\begin{align*}
(\exp^{-1} \circ L_g^{-1}) \times \text{id} : gV \times \mathfrak{g} \rightarrow U \times \mathfrak{g}.
\end{align*}
Because $g$ is fixed under $H$, this diffeomorphism intertwines the $H$-action. Consequently, $\exp^{-1} \circ L_g^{-1}$ maps $(gV \times \mathfrak{g}) \cap (G \times \mathfrak{g})^H$ onto
\begin{align*}
U^H \times \mathfrak{g}^H \subset \mathfrak{t} \times \mathfrak{t} \subset \mathfrak{g} \times \mathfrak{g}.
\end{align*}
Choose an orthonormal basis $\{e_i\}_{i=1}^t$ of $\mathfrak{t}$ such that $\{e_1, \dots , e_{t-k}\}$ spans $\mathfrak{g}^H \subset \mathfrak{t}$, where $k \geq 1$. By re-ordering the coordinates induced by the orthonormal basis $\{e_1,\dots, e_t\}$, the subset $(G \times \mathfrak{g})^H\cap   ( (gV  \cap \textbf{T}) \times \mathfrak{t})$ is mapped onto 
\begin{align*}
O^H:=\{ (x_1, \dots , x_{2t}) \in (U \cap \mathfrak{t}) \times \mathfrak{t} \mid x_{2t -2k +1} = \cdots x_{2t} = 0\} \subset \mathfrak{t} \times \mathfrak{t}
\end{align*}
 under the coordinate chart 
\begin{align*}
(\exp^{-1} \circ L_g^{-1}) \times \text{id} : (gV  \cap \textbf{T}) \times \mathfrak{t} \rightarrow (U \cap \mathfrak{t}) \times \mathfrak{t}
\end{align*}
for $\textbf{T} \times \mathfrak{t}$.  Under this chart the K\"ahler structure on $\textbf{T} \times \mathfrak{t}$ corresponds to the standard K\"ahler structure on $O:=(U \cap \mathfrak{t}) \times \mathfrak{t} \subset \mathbb{R}^{2t}$, because the differential of $\exp: \mathfrak{t} \rightarrow \textbf{T}$ is trivial when the tangent spaces on $\textbf{T}$ are identified with $\mathfrak{t}$ through left-trivialisation.
Because $G^H$ is compact, one can choose finitely many elements $g_l \in G^H$ such that $\{(g_l V \cap \textbf{T}) \times \mathfrak{t}\}$ covers $(G \times \mathfrak{g})^H = G^H \times \mathfrak{g}^H$. Then 
\begin{align*}
\mathcal{U} := \{(g_l V \cap \textbf{T} )\times \mathfrak{t}\} \cup \{\textbf{T} \times \mathfrak{t} \setminus (G \times \mathfrak{g})^H\}
\end{align*}
 is a finite open cover of $\textbf{T} \times \mathfrak{t}$. Consider the chart 
\begin{align*}
(\exp^{-1} \circ L_{g_l^{-1}}) \times \text{id}: (g_l V \cap \textbf{T}) \times \mathfrak{t} \rightarrow O
\end{align*}
for some fixed $l$.
Let $f \in \Gamma_c^\infty(O,\mathbb{E})$. When $O$ is regarded as an open subset of $\mathbb{R}^{2t}$, $f$ can be extended (by zero) to a section in $\Gamma^\infty_c(\mathbb{R}^{2t},\mathbb{E})$. Recall that the chart $(\exp^{-1} \circ L_{g^{-1}}) \times \text{id}$ maps the K\"ahler structure on $\textbf{T} \times \mathfrak{t}$ to the standard K\"ahler structure on $O=(U \cap \mathfrak{t}) \times \mathfrak{t}$. Therefore, the operator $D$ on $\Gamma^\infty_c(O,\mathbb{E})$ is just the restriction of the ordinary Dolbeault--Dirac operator $\widetilde{D}$ on $\mathbb{R}^{2t}$ to $O$. Note that for compactly supported sections on $O$ the graph norm with respect to $D$ is the same as the graph norm with respect to $\widetilde{D}$, as the operator $\widetilde{D}$ is local. 

By Corollary \ref{cor:denseDD}, a section $s \in \Gamma_c^\infty(O,\mathbb{E})$ can be approximated in the graph norm of $\widetilde{D}$ by a sequence $(s_m)_m \in \Gamma_c^\infty(\mathbb{R}^{2t} \setminus \mathbb{R}^{2t-2k},\mathbb{E})$. Let $\psi: \mathbb{R}^{2t} \rightarrow [0,1]$ be a smooth function with compact support contained in $O$ such that $\psi \equiv1$ on $\text{supp }s$. By Lemma \ref{lem:samesupportconvergence} below, $\psi s_m \rightarrow s$ in  the graph norm of $\widetilde{D}$. But 
\begin{align*}
\text{supp}(\psi s_m) \subset \text{supp}(\psi) \cap \text{supp}(s_m) \subset O \cap (\mathbb{R}^{2t} \setminus \mathbb{R}^{2t-2k}),
\end{align*}
 so that $\psi s_m \in \Gamma^\infty_c(O \setminus O^H, \mathbb{E})$ and $\psi s_m \rightarrow s$ in the graph norm of $D$. We have now proved that, for arbitrary $l$, $\Gamma^\infty_c( (g_l V \cap \textbf{T}) \times \mathfrak{t} \setminus (G \times \mathfrak{g})^H,\mathbb{E})$ is dense in $\Gamma^\infty_c((g_l V \cap \textbf{T}) \times \mathfrak{t},\mathbb{E})$.

Suppose now that $s \in \Gamma_c^\infty(\textbf{T} \times \mathfrak{t},\mathbb{E})$. Let $\{\psi_{U_i}\}_{U_i \in \mathcal{U}}$ be a partition of unity subordinate to the finite cover $\mathcal{U}$. The supports of the $\psi_{U_i}$'s are not compact, but the support of each $\psi_{U_i} s$ is, as $\text{supp}(\psi_{U_i} s) \subset \text{supp}(\psi_{U_i}) \cap \text{supp}(s)$ is a closed subset of the compact set $\text{supp}(s)$. Moreover, $\text{supp}(\psi_{U_i} s)$ is contained in $U_i$. By the previous paragraphs, each $\psi_{U_i} s$ can be approximated in the graph norm of $D$ by a sequence $(s_{m,i})_m \in \Gamma^\infty_c( \textbf{T} \times \mathfrak{t}, \mathbb{E})$ such that $\text{supp}(s_{m,i}) \subset U_i \setminus (U_i \cap (G \times \mathfrak{g})^H)$. 

Then $\sum_i s_{m,i} \rightarrow \sum_i \psi_{U_i} s = s$ as $m \rightarrow \infty$. Since the sum is finite, $\text{supp }(\sum_i s_{m,i})$ is compact and 
\begin{align*}
\text{supp  }\left(\sum_i s_{m,i} \right) \subset \bigcup_i \left( U_i \setminus (U_i \cap (G \times \mathfrak{g})^H) \right) = (\textbf{T} \times \mathfrak{t}) \setminus (G \times \mathfrak{g})^H
\end{align*}
 for each $m$. Thus  $\Gamma^\infty_c((\textbf{T} \times \mathfrak{t}) \setminus (G \times \mathfrak{g})^H,\mathbb{E})$ is dense in $\Gamma^\infty_c(\textbf{T} \times \mathfrak{t},\mathbb{E})$.
\end{proof}
The following lemma was used in the above proof.
\begin{lem}
\label{lem:samesupportconvergence}
Let $M$ be an oriented Riemannian manifold. Suppose that $(s_m)_m$ is a sequence in $\Gamma^\infty_c(M,E)$ such that $s_m \rightarrow s$ in $\Gamma^\infty_c(M,E)$ with respect to the graph norm of a first-order differential operator $D$ on a hermitian vector bundle $E$, and suppose that $\psi: M \rightarrow [0,1]$ is a compactly supported smooth function such that $\psi \equiv 1$ on $\text{supp }s$. Then $\psi s_m \rightarrow s$ with respect to the graph norm of $D$.
\end{lem}
\begin{proof}
We verify that
\begin{align*}
 \| s - \psi s_m  \|^2_D &= \| \psi(s - s_m)\|^2_D = \| \psi (s - s_m) \|^2 + \| D (\psi(s - s_m))\|^2 \\ &\leq \|s - s_m\|^2 + \| [D,\psi](s-s_m) + \psi D(s - s_m)\|^2 \\ &\leq \|s- s_m\|^2 + \left( \|[D,\psi]\| \|s- s_m\|  + \| D(s - s_m) \|\right)^2 ,
\end{align*}
 which goes to zero as $m \rightarrow \infty$, because $[D,\psi]$ is a bounded operator and $s_m \rightarrow s$ with respect to $\|\cdot\|_D$. 
\end{proof}

The next lemma will be applied to $P$ (see text after Definition \ref{def:QMWQ})
to prove that $\Gamma^\infty_c( (G \times \mathfrak{g})_\textbf{T},\mathbb{E})$ is dense in $\Gamma^\infty_c( \textbf{T} \times \mathfrak{t},\mathbb{E})$ with respect to the graph norm of $D$. 
\begin{lem}
\label{lem:removalmultiplestrata}
 Let $M$ be an oriented Riemannian manifold and let $D$ be a first-order differential operator on a hermitian vector bundle $E$ over $M$. Let $A_{1,2}$ be two closed subsets of $M$ such that $\Gamma^\infty_c(M \setminus A_i,E)$, ($i=1,2$), is dense in $\Gamma^\infty_c(M,E)$ with respect to the graph norm of $D$. Then $\Gamma^\infty_c(M \setminus (A_1 \cup A_2),E)$ is dense in $\Gamma_c^\infty(M,E)$ in this graph norm as well. 
\end{lem}
\begin{proof}
It suffices to show that $\Gamma^\infty_c(M \setminus (A_1 \cup A_2),E)$ is dense in $\Gamma^\infty_c(M \setminus A_1,E)$. Let $s \in \Gamma^\infty_c(M \setminus A_1,E)$ be given. Since $\Gamma^\infty_c(M \setminus A_2,E)$ is dense in $\Gamma^\infty_c(M,E)$ with respect to the graph norm of $D$, there exists a sequence $(s_m)_m$ in $\Gamma_c^\infty(M \setminus A_2,E)$ such that $s_m \rightarrow s$ in this norm. Now, let $\psi: M \rightarrow [0,1]$ be a function that has compact support contained in $M \setminus A_1$ and  is equal to $1$ on $\text{supp }s$. By Lemma \ref{lem:samesupportconvergence},
we see that  $\psi s_m \rightarrow s$ in the graph norm of $D$. But 
\begin{align*}
\text{supp }(\psi s_m) \subset \text{supp }\psi \cap  \text{supp }s_m \subset (M \setminus A_1) \cap (M\setminus A_2)  = M \setminus (A_1 \cup A_2)
\end{align*}
for each $m$.
\end{proof}
\begin{prp}
\label{prp:princstratess}
The space $\Gamma^\infty_c((G \times \mathfrak{g})_\textbf{T},\mathbb{E})$ is dense in $\Gamma^\infty_c(\textbf{T} \times \mathfrak{t},\mathbb{E})$ with respect to the graph norm of $D$. Moreover, $D$ is essentially self-adjoint on the domain $\Gamma^\infty_c((G \times \mathfrak{g})_\textbf{T},\mathbb{E})$.
\end{prp}
\begin{proof}
The union defining $P$ is finite. Proposition \ref{prp:denseMH} and Lemma \ref{lem:removalmultiplestrata} now imply that $\Gamma_c^\infty( (\textbf{T} \times \mathfrak{t}) \setminus P,\mathbb{E}) = \Gamma^\infty_c((G \times \mathfrak{g})_\textbf{T},\mathbb{E})$ is dense in $\Gamma^\infty_c(\textbf{T} \times \mathfrak{t},\mathbb{E})$ in the graph norm of $D$.
Concerning essential self-adjointness, first note that $D$ is essentially self-adjoint on $\Gamma_c^\infty(\textbf{T} \times \mathfrak{t},\mathbb{E})$, as $\textbf{T} \times \mathfrak{t}$ is geodesically complete and $D$ has finite propagation speed (see \cite[Proposition 10.2.11]{HR00}). Since $\Gamma^\infty_c((G\times \mathfrak{g})_\textbf{T},\mathbb{E})$ is dense in $\Gamma^\infty_c(\textbf{T}\times \mathfrak{t},\mathbb{E})$ with respect to the graph norm of $D$, $D$ is also essentially self-adjoint on the domain $\Gamma^\infty_c((G\times \mathfrak{g})_\textbf{T},\mathbb{E})$.
\end{proof}

Proposition \ref{prp:princstratess} deals with the ordinary Dolbeault--Dirac operator on $\textbf{T} \times \mathfrak{t}$. However, quantization is defined in terms of the \emph{twisted} Dolbeault--Dirac operator.  Recall that for cotangent bundles the twisting line bundle is $L$ is the trivial hermitian line bundle with hermitian connection $\nabla^L = d + 2\pi i \theta$, where $\theta$ is the fundamental $1$-form. Therefore, the twisted Dolbeault--Dirac operator differs from the untwisted one by a zeroth-order differential operator only.

\begin{prp} 
\label{prp:cotangent_torus}
The domain $\Gamma^\infty_c( (G \times \mathfrak{g})_\textbf{T},\mathbb{E} \otimes L)$ is dense in $\Gamma_c^\infty(\textbf{T} \times \mathfrak{t},\mathbb{E} \otimes L)$ in the graph norm of $D^L$. In particular, the twisted Dolbeault--Dirac operator $D^L$ is essentially self-adjoint on the domain $\Gamma^\infty_c((G \times \mathfrak{g})_\textbf{T},\mathbb{E} \otimes L)$.
\end{prp}
\begin{proof}
Since $L$ is trivial and has the standard hermitian structure, one can identify $\mathbb{E} \otimes L$ and $\mathbb{E}$ as hermitian vector bundles. As a differential operator on $\mathbb{E}$, the difference $D^L - D$ is of order zero. In particular, $D^L$ has the same principal symbol as $D$, so that $D^L$ still has finite propagation speed and is therefore still essentially self-adjoint on $\Gamma_c^\infty(\textbf{T} \times \mathfrak{t},\mathbb{E})$. It remains to show that $\Gamma_c^\infty((G \times \mathfrak{g})_\textbf{T},\mathbb{E})$ is dense in $\Gamma^\infty_c(\textbf{T} \times \mathfrak{t},\mathbb{E})$ in the graph norm of $D^L$. For simplicity we denote $D^L = D + B$, where $B$ is an element of $\Gamma^\infty(\textbf{T} \times \mathfrak{t}, \End(\mathbb{E}))$. 

Suppose that $s \in \Gamma_c^\infty(\textbf{T} \times \mathfrak{t},\mathbb{E})$. By Proposition \ref{prp:princstratess} there exists a sequence $(s_m)_m$ in $\Gamma_c^\infty((G \times \mathfrak{g})_\textbf{T},\mathbb{E})$ such that $s_m \rightarrow s$ in the graph norm of $D$. According to Lemma \ref{lem:samesupportconvergence} one can assume that there exists a compact set $K$ such that $(\cup_m \text{supp }s_m) \cup \text{supp }s \subset K$. Consequently,
\begin{align*}
 \| (D+B)(s_m -s)\| &\leq \|D(s_m -s)\| + \|B(s_m - s)\| \\ &\leq \| D(s_m -s)\| + \sup_{x \in K}\{ \|B(x)\| \} \|s_m -s \|,
\end{align*}
which approaches $0$ as $m$ goes to infinity. So $\Gamma^\infty_c(G \times \mathfrak{g})_\textbf{T},\mathbb{E})$ is dense in $\Gamma_c^\infty(\textbf{T} \times \mathfrak{t},\mathbb{E})$ in the graph norm of $D^L=D+B$.
\end{proof}

As an immediate consequence of Proposition \ref{prp:cotangent_torus}, one has
\begin{cor}
\label{cor:cotangent_torus}
The inclusion 
\begin{align*}
 \iota: \Gamma^\infty_c((G\times\mathfrak{g})_\textbf{T}, \mathbb{E} \otimes L) \rightarrow \Gamma^\infty_c(\textbf{T} \times \mathfrak{t}, \mathbb{E} \otimes L)
\end{align*}
extends to an identification of 
\begin{align*}
 \iota: L^2((G \times \mathfrak{g})_\textbf{T},\mathbb{E} \otimes L) \rightarrow L^2(\textbf{T} \otimes \mathfrak{t}, \mathbb{E} \otimes L).
\end{align*}
Moreover, if $D^L_{princ}$ denotes the twisted Dolbeault--Dirac operator on $(G \times \mathfrak{g})_\textbf{T}$  and $D^L$ is the twisted Dolbeault--Dirac operator on $\textbf{T} \times \mathfrak{t}$, then 
\begin{align*}
 \ol{D}^L_{princ} = \ol{D}^L,
\end{align*}
as operators on $L^2(\textbf{T} \times \mathfrak{t}, \mathbb{E} \otimes L)$, where the identification of the Hilbert spaces through $\iota$ is implicit. In particular,
\begin{align*}
 \ker \ol{D}_{princ} = \ker \ol{D}^L.
\end{align*}
\end{cor}

To study the Dolbeault--Dirac quantization on $(G \times \mathfrak{g})_\textbf{T} / W(G,\textbf{T})$ we first need some other facts. The proof of the following lemma is straightforward. 
\begin{lem}
\label{lem:essentiallyselfadjointprojection}
 Let $D: \mathcal{D}(D) \rightarrow \mathcal{H}$ be a closable operator, and denote the closure of $D$ by $\ol{D}$. Let $p \in \mathcal{B}(\mathcal{H})$ be a projection such that $p (\mathcal{D}(D)) \subset \mathcal{D}(D)$ and $pD = Dp$ on $\mathcal{D}(D)$. Then $D$ restricts to a densely defined closable operator $D_p: p(\mathcal{D}(D)) \rightarrow p\mathcal{H}$ on $p\mathcal{H}$. 
Moreover, $p\ol{D}= \ol{D}p$ on $\mathcal{H}$ and $\ol{D_p} = \ol{D}|_{p\mathcal{H}}$ on $p\mathcal{H}$. If $D$ is essentially self-adjoint, then so is $D_p$.
\end{lem}

If a Lie group acts on an oriented Riemannian manifold, we always assume that the action preserves the metric as well as the orientation.
\begin{prp}
\label{prp:esap}
 Let $\Gamma$ be a \emph{finite} group acting on an arbitrary oriented Riemannian manifold $M$. Suppose that $D$ is a symmetric $\Gamma$-invariant differential operator on a $\Gamma$-equivariant hermitian vector bundle $E$. Then 
\begin{align*}
\ol{D}|_{L^2(M,E)^\Gamma} = \ol{D|_{L^2(M,E)^\Gamma}}.
\end{align*}
Moreover, if $D$ is essentially self-adjoint on the domain $\Gamma_c^\infty(M,E)$, then $D|_{L^2(M,E)^\Gamma}$ is essentially self-adjoint on the domain $\Gamma_c^\infty(M,E)^\Gamma$.
\end{prp}
\begin{proof}
 The inclusion of the closed subspace $L^2(M,E)^\Gamma$ into $L^2(M,E)$ has a left-inverse $p: L^2(M,E) \rightarrow L^2(M,E)^\Gamma$, given by
\begin{align*}
ps(x) = \frac{1}{|\Gamma|} \sum_{g \in\Gamma}  g(s(g^{-1}x)), \quad (x \in M).
\end{align*}
The map $p$ is easily verified to be the projection onto $L^2(M,E)^\Gamma$. Furthermore, $p(\Gamma^\infty_c(M,E) ) = \Gamma^\infty_c(M,E)^\Gamma$. Since $D$ is assumed to be $\Gamma$-invariant, it commutes with the projection $p$. Now apply Lemma \ref{lem:essentiallyselfadjointprojection}.
\end{proof}

Suppose that $\Gamma$ acts freely on $M$ and write $\pi: M \rightarrow M/\Gamma$ for the quotient map. Let $E$ be a $\Gamma$-equivariant vector bundle over $M$. Then $E/\Gamma$ is a vector bundle on $M/\Gamma$ with the obvious projection map $E/ \Gamma \rightarrow M/ \Gamma$. Moreover, the map sending a section $s \in \Gamma(M,E)^\Gamma$ to the section $\widetilde{s} \in \Gamma^\infty(M/\Gamma, E/\Gamma)$ that is given by
\begin{align*}
 \widetilde{s}([x]) := [s(x)], \quad (x \in M),
\end{align*}
 is an isomorphism of $C^\infty(M/\Gamma) \cong C^\infty(M)^\Gamma$-modules.

Because $\Gamma$ is discrete, there is natural identification of $(TM)/\Gamma$ with $T(M/\Gamma)$. Moreover, if $M$ is endowed with a $\Gamma$-invariant K\"ahler structure, then $M/\Gamma$ inherits a K\"ahler structure from $M$, and if $L$ is a $\Gamma$-equivariant pre-quantum line bundle over $M$ for the symplectic structure $\omega$ on $M$, then $L/\Gamma$ is a pre-quantum line bundle over $M/ \Gamma$ (with connection $\nabla^{L/\Gamma}$ induced by the $\Gamma$-invariant connection $\nabla^L$ on $L$) for the inherited symplectic structure $\omega_\Gamma$ on $M/\Gamma$.  In fact, the pull-back of $\omega_\Gamma$ along $\pi$ is $\omega$, and the pull-back of $(L/\Gamma, \nabla^{L/\Gamma})$ is $(L,\nabla^L)$.

Because $\pi: M \rightarrow M / \Gamma$ is a covering map, $M$ can be covered by open subsets $U_i$ on which $\pi$ is a diffeomorphism onto the open subset $\pi(U_i)$ in $M / \Gamma$. The map $\pi$ then identifies the K\"ahler structure on $U_i$ with the one on $\pi(U_i)$. It follows that the isomorphism
\begin{align*}
\Gamma_c^\infty(M,\mathbb{E}\otimes L)^\Gamma \rightarrow \Gamma_c^\infty(M/\Gamma, \mathbb{E}_{M/\Gamma} \otimes L/\Gamma), \quad  s \mapsto \widetilde{s}
\end{align*}
intertwines the Dolbeault--Dirac operators on $M$ and $M/\Gamma$.

Let $\varepsilon$ and $\widetilde{\varepsilon}$ be the Liouville measures on $M$ and $M/\Gamma$, respectively. Then 
\begin{align*}
 \int_M (\pi^*\tilde{f}) \varepsilon = |\Gamma| \int_{M/\Gamma} \tilde{f} \tilde{\varepsilon},
\end{align*}
for all $\widetilde{f} \in C^\infty(M/\Gamma)$. Therefore, for any $\Gamma$-equivariant hermitian vector bundle $E \rightarrow M$, the map $u: \Gamma^\infty_c(M,E)^\Gamma \rightarrow \Gamma^\infty_c(M/\Gamma,E/\Gamma)$ given by
\begin{align*}
u: s\mapsto |\Gamma|^{\frac{1}{2}} \widetilde{s}, \quad s \in \Gamma^\infty_c(M,E)^\Gamma
\end{align*}
is unitary. Taking $E = \mathbb{E} \otimes L$ and using Proposition \ref{prp:esap}, we now obtain:
\begin{prp}
\label{prp:discretequotient}
Let $\Gamma$ be a finite group acting \emph{freely} on a K\"ahler manifold $M$, such that the K\"ahler structure is $\Gamma$-invariant. Suppose that $(L,\nabla^L)$ is an equivariant pre-quantization. Then the map
\begin{align*}
u: L^2(M,\mathbb{E} \otimes L)^\Gamma \rightarrow L^2(M/\Gamma,\mathbb{E}_{M/\Gamma} \otimes L/\Gamma), \quad s^\Gamma \mapsto |\Gamma|^\frac{1}{2} \widetilde{s}
\end{align*}
is a unitary isomorphism that intertwines the Dolbeault--Dirac operators $D^L$ and $D^{L/\Gamma}$. Moreover,
with $\mathcal{H} = L^2(M,\mathbb{E} \otimes L)$ we have
\begin{align*}
( \ker \ol{D}^L_+)^\Gamma &= \ker ( ( \ol{D}^L_+ )|_{\mathcal{H}^\Gamma}) =   \ker\ol{(D^L_+)|_{\mathcal{H}^\Gamma}} \stackrel{u}{\cong} \ker \ol{D}^{L/\Gamma}_+, \\
( \ker \ol{D}^L_-)^\Gamma &= \ker (( \ol{D}^L_- )|_{\mathcal{H}^\Gamma} ) =   \ker \ol{(D^L_-)|_{\mathcal{H}^\Gamma}} \stackrel{u}{\cong} \ker \ol{D}^{L/\Gamma}_-.
\end{align*}
\end{prp}

Let us apply Proposition \ref{prp:discretequotient} to our situation, where the finite group $\Gamma := W(G,\textbf{T})$ acts freely on $M = (G\times \mathfrak{g})_\textbf{T}$.  Composition of the map $u^{-1}$ of Proposition \ref{prp:discretequotient} and $\iota$ in Corollary \ref{cor:cotangent_torus} yields an identification
\begin{align*}
 \iota \circ u^{-1}: L^2((G\times\mathfrak{g})_\textbf{T} / \Gamma, \mathbb{E}_{(G\times \mathfrak{g})_\textbf{T}/\Gamma} \otimes L/ \Gamma)  \stackrel{\cong}{\rightarrow} L^2(\textbf{T} \times \mathfrak{t}, \mathbb{E} \otimes L) ^\Gamma.
\end{align*}
If $D^{L/\Gamma}$ denotes the Dolbeault--Dirac operator on $(G \times \mathfrak{g})_\textbf{T} / \Gamma$ and if $D^L$ now denotes the Dolbeault--Dirac operator on $\textbf{T} \times \mathfrak{t}$, then the isomorphism $\iota \circ u^{-1}$ restricts to isomorphisms
\begin{align}
\label{eq:kern_id}
\iota \circ u^{-1}: \ker \ol{D}_+^{L/\Gamma} &\stackrel{\cong}{\rightarrow} (\ker \ol{D}^L_+)^\Gamma, \\
\iota \circ u^{-1}: \ker \ol{D}_-^{L/\Gamma} &\stackrel{\cong}{\rightarrow} (\ker \ol{D}^L_-)^\Gamma. \nonumber
\end{align}
From (the proof of) Theorem \ref{thm:quantDD_cot}, we know that $\ker \ol{D}^L_- = \{0\}$ and that $\ker \ol{D}_+^L$ is concentrated in degree $0$ on $\textbf{T} \times \mathfrak{t}$. By the above results, the same  is true for the Dolbeault--Dirac operators on $(G\times \mathfrak{g})_\textbf{T}$ and on the principal stratum $(G\times \mathfrak{g})_\textbf{T} / \Gamma$ of the Marsden-Weinstein quotient $j^{-1}(0) / \Ad G$. Thus we are in the setting of Remark \ref{rmk:quantization}, so Definition \ref{dfn:quantization} of Dolbeault--Dirac quantization again remains close to the definition of quantization as an index.

This brings us to one of our main results:
\begin{thm}
\label{thm:princstrateq}
 Defining the Dolbeault--Dirac quantization of $j^{-1}(0) /\Ad G$ as the Dolbeault--Dirac quantization of its principal stratum, one has
\begin{align*}
 \mathcal{Q}_{DD}(j^{-1}(0) /\Ad G) :=  \ker \ol{D}_+^{L/W(G,\textbf{T})} \cong \mathcal{Q}_{DD}(T^*\textbf{T})^{W(G,\textbf{T})}\cong L^2(\textbf{T})^{W(G,\textbf{T})},
\end{align*}
where each unitary isomorphism is natural (i.e., independent of a choice of basis).
\end{thm}

\begin{proof}
By definition, $\mathcal{Q}_{DD}(T^*\textbf{T}):=\ker \ol{D}_+^L$. The first isomorphism is then a consequence of (\ref{eq:kern_id}). The second follows from Theorem \ref{thm:quantspin} and Proposition \ref{prp:Hall_iso}.
 \end{proof}

\begin{rmk}
\label{rmk:princstrateq}
Rephrased in terms of holomorphic sections, Proposition \ref{prp:discretequotient} states that the square-integrable holomorphic sections of the bundle $L/W(G,\textbf{T})$ on the principal stratum $(G\times\mathfrak{g})_\textbf{T} / W(G,\textbf{T})$ can be identified with the $W(G,\textbf{T})$-invariant square-integrable holomorphic sections of the $W(G,\textbf{T})$-equivariant bundle $L$ on $(G \times \mathfrak{g})_\textbf{T}$. Corollary \ref{cor:cotangent_torus}, on the other hand, says that \emph{square-integrable} holomorphic sections  of $L$ on the open, dense submanifold $(G \times \mathfrak{g})_\textbf{T} \subset \textbf{T} \times \mathfrak{t}$ can always be extended to a square-integrable \emph{holomorphic} section on the entire manifold $\textbf{T} \times \mathfrak{t}$. The latter result explains  why the quantization of the principal stratum is equal to the quantization of the full manifold. Since in general holomorphic sections on a dense, open neighborhood do not always have holomorphic extensions to the full space, this extension-result is decidedly non-trivial. It relies on the facts that the codimensions of the other strata are at least $2$ and that the sections are square-integrable. It would be worthwhile to investigate if a similar results holds in a more general setting of symplectic stratitification, where codimensions of the strata other than the principal one are automatically $2$ or higher. The definition of quantization of the singular space as the quantization of its principal stratum might then be extented to this setting as well. 
\end{rmk}

\subsection{Quantization commutes with reduction}
 We discuss some facts concerning the Guillemin-Sternberg conjecture for the coadjoint action of $G$ on $T^*G$. The following proposition is well known.
\begin{prp}
\label{prp:qar}
Let $G$ be a compact connected Lie group and $\textbf{T}$ a maximal torus. Write $\delta: \textbf{T} \rightarrow \mathbb{C}$ for the Weyl denominator function. Then there exists $c>0$ such that $f \mapsto  c|\delta| \cdot f|_\textbf{T}$ defines a unitary map
\begin{align*}
 L^2(G)^{\Ad G} \rightarrow L^2(\textbf{T})^{W(G,\textbf{T})}.
\end{align*}
\end{prp}
\begin{proof}
See e.g. \cite[Corollary 3.14.2]{DK00}.
\end{proof}
In Theorem \ref{thm:quantspin} we showed that the Dolbeault--Dirac quantization of $T^*G$ with its standard K\"ahler structure is $\mathcal{H}L^2(T^*G, e^{-|Y|^2} \varepsilon)$, which is $G \times G$-equivariantly isomorphic to $L^2(G)$ via the isomorphism of Proposition \ref{prp:Hall_iso}. This isomorphism composed with the Weyl integration formula produces a canonical isomorphism between the reduction after quantization and $L^2(\textbf{T})^{W(G,\textbf{T})}$.
On the other hand, in Theorem \ref{thm:princstrateq} we applied the isomorphism $u$ in Proposition \ref{prp:discretequotient}, and again Proposition \ref{prp:Hall_iso} to also identify $\mathcal{Q}_{DD}(T^*G /\!/ \Ad G)$ with $L^2(\textbf{T})^{W(G,\textbf{T})}$. Hence:
\begin{thm}
\label{thm:qr0}
Defining the Dolbeault--Dirac quantization of $T^*G /\!/ \Ad G$ to be the Dolbeault--Dirac quantization of its principal stratum,  (Dolbeault--Dirac) quantization after reduction and reduction after quantization are both canonically isomorphic to $L^2(\textbf{T})^{W(G,\textbf{T})}$, and hence quantization commutes with reduction.
\end{thm}
\section{Discussion and outlook}\label{Outlook}
To summarize the main results in this paper, we have:
\begin{itemize}
\item  Shown that the Dolbeault--Dirac  quantization of $T^*G$ (with its standard K\"ahler structure) yields the same Hilbert space (\ref{eq:Hallq})  that Hall found in \cite{Hal02} (who used geometric quantization based on a holomorphic polarisation). The main point we proved in this light  is the fact that the kernel of $D^L_-$ is trivial and that the kernel of $D^L_+$ is precisely the space of holomorphic square-integrable sections of $L$ (Theorem \ref{thm:quantDD_cot});
\medskip

\item Formulated a quantization procedure for the singular Marsden--Weinstein quotient $T^*G /\!/ \Ad G$, where quantization is performed by taking the kernel of the (twisted) Dolbeault--Dirac operator on the principal stratum. It turned out that it is sufficient to consider the principal stratum because the Dolbeault--Dirac operator is still essentially self-adjoint there (Propositions \ref{prp:princstratess} and \ref{prp:discretequotient}). The essential self-adjointness relies on the fact that the singular strata are all of a codimension greater than 2;\medskip

\item Shown that quantization commutes with reduction in the sense that both reduction after quantization and quantization after reduction yield the same Hilbert space $L^2(\textbf{T})^{W(G,\textbf{T})}$ (Theorems \ref{thm:quantspin}, \ref{thm:princstrateq} and Proposition \ref{prp:qar}).  
\end{itemize}
\medskip

Although  the above methods for quantizing the cotangent bundle $T^*G$ and the Marsden--Weinstein reduced space $T^*G /\!/ \Ad G$ are perfectly natural, the `quantization commutes with reduction theorem' would get more body if there were a way to identify quantization after reduction with reduction after quantization differently from mere unitary isomorphism of Hilbert spaces. Now there is a natural map  $$\mathcal{H}L^2(T^*G, e^{-|Y|^2}\varepsilon)^{\Ad G}\rightarrow \mathcal{H}L^2(  (j^{-1}(0)/ \Ad G)_{princ} , e^{-|\tilde{Y}|^2} \widetilde{\varepsilon}),$$ where a holomorphic function $$f \in \mathcal{H}L^2(T^*G, e^{-|Y|^2}\varepsilon)^{\Ad G}$$ is pulled back to the space  $(G\times \mathfrak{g})_\textbf{T} \subset T^*\textbf{T}$ and then projected to a square-integrable holomorphic function on the quotient $(G \times \mathfrak{g})_\textbf{T} / W(G,\textbf{T})$, but unfortunately this map is not unitary \cite{HK07}. 
To solve this problem one could look for 
either other natural maps going from the one space to the other, or  some more flexible  \emph{framework} in which the quantization-commutes-with-reduction problem can be formulated. 
Unfortunately, neither the $K$-theoretic framework for quantization-commutes-with-reduction  used in \cite{Lan05,HL08,HM15,MZ10}
nor the approaches in  \cite{MZ09,Par11} apply here, since our momentum map fails to be proper and the reduced space is non-compact, so that
the quantization of the latter, viz.\ the Hilbert space  $L^2(G)^{\Ad G}$, is infinite-dimensional, so that, in particular, the multiplicity of the trivial representation in the $\Ad G$-equivariant quantization of $T^*G$ is infinite. Consequently,  $L^2(G)$ can neither be interpreted as an element in the generalised representation ring, nor as an element in $K(C^*(\{e\})) = K_0(\mathbb{C}) \cong \mathbb{Z}$.
 It would  therefore be desirable to find a new framework in which the `quantization commutes with reduction' problem can be studied for singular cases like the one studied in this paper.
 \newpage
 \section*{Acknowledgements}
The research of the first author was funded by the Radboud University Nijmegen and  the \emph{Netherlands Organisation for Scientific Research}
 (NWO),  both through the \emph{Geometry and Quantum Theory} (GQT) cluster  and through a VIDI grant to the third author.
  Part of this work was  carried out during  visits
 by the first author  to the \emph{Centre for Symmetry and Deformation} at the University of Copenhagen and the \emph{Institute for Analysis} at the Leibniz Universit\"at Hannover. JB's  visit to Copenhagen was funded by the \emph{European Science Foundation} (ESF) as part of its program \emph{Interactions of Low-Dimensional Topology and Geometry with Mathematical Physics} (ITGP). His visit to Hannover was funded by the host institute. 
  
We would like to thank Ryszard Nest and Elmar Schrohe for their support and hospitality, and we are also
 grateful to Heiko Gimperlein, Gerd Grubb, Brian Hall, Eli Hawkins, Nigel Higson, Peter Hochs, Johannes Huebschmann, William Kirwin,  Reyer Sjamaar, and last but not least the outstanding  referee for \emph{Reviews in Mathematical Physics} for very helpful comments and suggestions.
\bigskip
 
\noindent \textbf{Note added in proof}. This paper  was originally posted as \verb#arXiv:1508.06763#. While it was under review, an important contribution to the subject appeared, namely  ``A unitary `quantization commutes with reduction' map for the adjoint action of a compact Lie group" by B.C. Hall and B.D. Lewis (\verb#arXiv:1709.08531#), subsequently published in \emph{The Quarterly Journal of Mathematics} 00, 1--35 (2018). This paper addresses the same problem as ours, but instead of using an index-theoretic definition of quantization (which is really K-theoretic, \textit{cf.}\ the Introduction), they use the traditional geometric quantization approach, amended with the half-form correction \cite{Hal2013,Woodhouse}. In an impressive \emph{tour de force}, Hall and Lewis  succeed in proving a `quantization commutes with reduction' result, too, including a (geometrically) more natural unitary map from the `quantization before reduction' Hilbert space to the `quantization after reduction' Hilbert space than ours. 

Let us therefore close by explaining the  connection of this paper to ours. 
In the original version of our paper we already tried to extend our result for the 
 Spin$\mbox{}^{\mathbb{C}}$--Dirac operator =  Dolbeault--Dirac operator (\textit{cf.}\ Remark \ref{5r}) to a possible Spin Dirac operator
 $\ol{\sD}$ on $M=T^*G$, in terms of which, analogous to Definition \ref{dfn:quantization}, the  \emph{spin quantization} of $M$ would be defined as the Hilbert space
\begin{equation}
 \mathcal{Q}^L_S(M) := \ker (\ol{\sD}_+^L). \label{spinDq}
\end{equation}
 It is shown in e.g. \cite[Chapter 6]{Dui96} that for a fixed spin structure on any K\"ahler manifold $M$, the spinor bundle $S$ is isomorphic to the bundle $\Lambda^\bullet(T^{*(0,1)}M) \otimes K_\frac{1}{2}$, where $K_\frac{1}{2}$ is the line bundle of half-forms corresponding to the chosen spin structure (which satisfies $K_\frac{1}{2} \otimes K_\frac{1}{2} = K$).  Furthermore,
the canonical line bundle on $T^*G$ is trivial. Let $\{\beta_i\}_{i=1}^n$ be a linearly independent system of left-invariant holomorphic $(1,0)$-forms on $G^\mathbb{C}$. Then $\beta := \beta_1 \wedge \cdots \wedge \beta_n$ is a left $G^\mathbb{C}$-invariant holomorphic trivialising section of $K$. The section $\beta$ is also invariant under the right action of $G$ on $G^\mathbb{C}$. To see this, note that 
\begin{align*}
 (T_h R_g)^*  \beta(hg) = \det( \Ad^*_{g^{-1}}) \beta(h), \quad (h \in G^\mathbb{C}, g \in G),
\end{align*}
where $\Ad^*_{g^{-1}}$ is viewed as a real-linear map $\mathfrak{g}^* \rightarrow \mathfrak{g}^*$. The function 
\begin{align*}
G \rightarrow \mathbb{R}^{\times}, \quad g \mapsto \det(\Ad^*_{g^{-1}})
\end{align*}
is a group homomorphism. Since $G$ is compact and connected, the image of this homomorphism is a compact and connected subgroup of $\mathbb{R}^\times$, hence this image is $\{1\}$. Consequently, $\beta$ is also invariant under the right action of $G$ on $G^\mathbb{C}$. Since $\beta$ is both left- and right- invariant, it is invariant under the action of $G \times G$ on $G^\mathbb{C}$.

Let $h$ denote the hermitian structure on $K$. Now choose $K_\frac{1}{2}$ to be the trivial holomorphic line bundle with the trivialising holomorphic section $\alpha$ that satisfies $\alpha^2 := \alpha \otimes \alpha= \beta$ and with hermitian structure determined by $h(\alpha , \alpha)= h( \beta,\beta)^\frac{1}{2}$,
so that $R(K_{\frac{1}{2}}) = \tfrac{1}{2} R(K)$. The group action on $K_{\frac{1}{2}}$ is such that $\alpha$ is $G \times G$-invariant.
\begin{thm}
\label{prp:quantspin_cot}
The spin quantization  $\ker \ol{\sD}^L$  of $T^*G$ (with its standard K\"ahler structure)
 is $G \times G$-equivariantly isomorphic to $\mathcal{H}L^2(T^*G, e^{-2\pi|Y|^2} \eta \varepsilon)$.
\end{thm}
\begin{proof}
The twisted operator $\sD^L$ is equal to the twisted Dolbeault--Dirac operator on $\Lambda^ \bullet(T^{*(0,1)}M) \otimes (K_\frac{1}{2} \otimes L)$ \cite{Dui96, Fri00}. Because the canonical line bundle $K$ is semi-negative, the line bundle $K^* \otimes (K_\frac{1}{2} \otimes L)$ is positive so that, by Theorem \eqref{thm:kvt}, the kernel of $\ol{\sD}^L$ is contained in $\Gamma^\infty(M, K_\frac{1}{2} \otimes L)$, and as in Theorem \ref{thm:quantDD_cot} this kernel is
\begin{align*}
 \ker \ol{\sD}^L &= \mathcal{H}L^2(T^*G,K_\frac{1}{2} \otimes L) \cong \mathcal{H}L^2(T^*G, e^{-2\pi|Y|^2} \eta \varepsilon),\end{align*}
where the last isomorphism is given by the map
\begin{align}
\label{eq:isometa}
\mathcal{H}L^2(T^*G,K_\frac{1}{2} \otimes L) \stackrel{\cong}{\rightarrow} \mathcal{H}L^2(T^*G, e^{-2\pi|Y|^2} \eta \varepsilon), \quad  fe^{-\pi |Y|^2}\alpha \mapsto f.
\end{align} 
This is an isomorphism, because the trivialising section $\alpha$ of $K_\frac{1}{2}$ satisfies 
$ h(\alpha,\alpha)= \eta$. 
Because $\alpha$ and $e^{-\pi|Y|^2}$ are invariant under $G \times G$, the isomorphism of \eqref{eq:isometa} is $G \times G$-equivariant.
\end{proof}

As in Theorem \ref{thm:quantspin}, the equivariant spin quantization  (\ref{spinDq}) of $T^*G$ is   $G \times G$-equivariantly  equal to (\ref{eq:Hallqhalf}), and hence by Hall's result
  \cite[Theorem 2.6]{Hal02} also this quantizations of $T^*G$ is $G \times G$-equivariantly isomorphic to $L^2(G)$. 
  
 Thus the Spin Dirac operator quantization of the unreduced space $T^*G$ is well under control.  
However, in attempting to pass to the reduced space we stumbled on the following problem.  The curvature of the canonical line bundle on $T^*\textbf{T}$ is $0$, so that the Dolbeault--Dirac quantization is equal to the spin quantization for $T^*\textbf{T}$. When we also consider the $W(G,\textbf{T})$-action on the canonical bundle, then the left-invariant holomorphic $(n,0)$-form $\beta_1 \wedge \cdots \wedge \beta_n$, where each $\beta_k$ is a left-invariant $(1,0)$-form on $\textbf{T}^\mathbb{C}$, is not $W(G,\textbf{T})$-invariant. Indeed, an element $w\in W(G,\textbf{T})$ sends $\beta_1 \wedge \cdots \wedge \beta_n$ to $\det(w) \beta_1 \wedge \cdots \wedge \beta_n$, where $\det_{\mathfrak{t}}(w)$ is the determinant of the action of $w$ as a real-linear map on $\mathfrak{t}$, or equivalently, the determinant of its complex linear extension to $\mathfrak{t}_\mathbb{C}$, considered as a complex-linear map. This determinant is $\pm 1$, depending on whether $w$ is a rotation or a reflection of $\mathfrak{t}$, and at the time we did not know  if there existed an equivariant half-form bundle whose square is  the equivariant canonical line bundle. Now Hall and Lewis (\S5) show that this bundle  indeed exists,  and in this light we conjecture that  quantization defined as the index of the ensuing Spin Dirac operator coincides with geometric quantization using the half-form correction, as constructed by Hall and Lewis. We do not see how to prove this, but if this conjecture is true, their result would show that, on the definition (\ref{spinDq}), at least if $G$ is simply connected, quantization commutes with reduction also for Spin Dirac operators (and not only for Spin$\mbox{}^{\mathbb{C}}$--Dirac operators, as we prove here).
\appendix
\section{Equivariance of Hall's isomorphisms}
\label{sct:Hall_iso}
In this appendix we show that Hall's isomorphism in \cite{Hal94} is equivariant. We first recall this isomorphism.
Define an  entire function $\phi: G^\mathbb{C} \rightarrow \mathbb{C}$  by
\begin{align*}
\phi(t) =  \sum_{\pi \in \hat{G}} \frac{\text{dim }V_\pi}{\sqrt{\sigma(\pi)}}\Tr(\pi(t^{-1})),
\end{align*}
where the sum is over all irreducible representations of $G$, each of which is extended holomorphically to an irreducible representation of $G^\mathbb{C}$, and where 
\begin{align*}
\sigma(\pi) = \frac{1}{\text{dim }V_\pi}\int_{G^\mathbb{C}} \|\pi(t^{-1}) \|^2 e^{-2 \pi|Y|^2} \varepsilon.
\end{align*}
Then by \cite[Theorem 10]{Hal94}, the following map is unitary:
\begin{eqnarray}
C_\phi: L^2(G) &\rightarrow &\mathcal{H}L^2(T^*G,e^{-2\pi|Y|^2} \varepsilon);\label{Halliso1}\\
(C_\phi f)(t) &=& \int_G f(x) \phi(x^{-1}t) dx, \quad (f \in L^2(G)).\label{Halliso2}
\end{eqnarray}
By restriction, $C_\phi$ also defines a unitary map
\begin{equation}
C_\phi': L^2(\textbf{T})\rightarrow L^2(T^*\textbf{T}, e^{-2\pi|Y|^2} \varepsilon),
\end{equation}
where $\textbf{T}$ is a maximal torus of $G$ and  $W(G,\textbf{T})$ is the corresponding Weyl group.
\begin{prp}
\label{prp:Hall_iso} Let $G$ be a compact connected Lie group. 
\begin{enumerate}
\item  Hall's unitary isomorphism $C_\phi$ intertwines the natural $G\times G$-actions.
\item   Its restriction $C_\phi'$
intertwines the pertinent $W(G,\textbf{T})$-actions.
\end{enumerate}
\end{prp}
\begin{proof}
(1) We apply \cite[Theorem 10]{Hal94}, where unitary isomorphisms are constructed between $L^2(G)$ and $\mathcal{H}L^2(G^\mathbb{C},\nu)$ for a specific class of $G$-bi-invariant measures $\nu$ on $G^\mathbb{C}$. To verify that this theorem applies to the $G\times G$-invariant measure $$\nu:= e^{-2\pi|Y|^2} \varepsilon,$$ we need to check the following:
\begin{enumerate}
\item $\nu$ is given by a density with respect to the Haar measure $d\mu$ on $G^\mathbb{C}$ that is locally bounded away from zero;
\item for each irreducible representation $\pi$ of $G$, analytically continued to $G^\mathbb{C}$, the expression
\begin{align*}
 \int_{G^\mathbb{C}} \|\pi(t)^{-1}\|^2 d\nu_t
\end{align*}
is finite. 
\end{enumerate}
By \cite[Lemma 5]{Hal97}, $d\mu = \eta^2 \varepsilon$, so that $$\nu = e^{-2\pi|Y|^2} \varepsilon = e^{-2\pi|Y|^2} / \eta^2 d\mu.$$ The function $Y\mapsto e^{-2\pi|Y|^2} / \eta^2$ is smooth and strictly positive, in particular it is locally bounded away from zero. Furthermore, since any irreducible representation of $G$ is finite-dimensional and hence can be assumed to be unitary, we have
\begin{align*}
\int_{G^\mathbb{C}} \| \pi(t)^{-1} \|^2 e^{-2\pi|Y|^2} \varepsilon &= \int_{G \times \mathfrak{g}} \| \pi(ge^{iY})^{-1} \|^2 e^{-2\pi|Y|^2} dxdY \\ &= \int_{G \times \mathfrak{g}} \| \pi(g^{-1})\pi(e^{iY})^{-1} \|^2 e^{-2\pi|Y|^2} dx dY \\ &= |G| \int_{\mathfrak{g}} \|\pi(e^{iY})^{-1} \|^2 e^{-2\pi|Y|^2} dY,
\end{align*}
where we used the fact that on $G \times \mathfrak{g}$ the Liouville measure is equal to $dxdY$, assuming that $dx$, the Haar measure on $G$,  and $dY$, the Lebesgue measure on $\mathfrak{g}$, are appropriately normalised (see \cite[Lemma 4]{Hal97}). Writing $\ol{\pi}$ for the induced representation of the Lie algebra, we obtain
\begin{align*}
\int_{G^\mathbb{C}} \| \pi(t)^{-1} \|^2 e^{-2\pi|Y|^2} \varepsilon &=|G| \int_{\mathfrak{g}} \|e^{-i \ol{\pi}(Y)}\|^2 e^{-2\pi|Y|^2} dY\\ &\leq |G| \int_\mathfrak{g} e^{2\|\ol{\pi}(Y)\|}e^{-2\pi|Y|^2} dY \\ & \leq |G|\int_\mathfrak{g} e^{C|Y|}e^{-2\pi|Y|^2} dY,
\end{align*}
for some constant $C>0$. This last integral is finite. Thus  the measure $\nu$ satisfies conditions (1) and (2), so we can apply \cite[Theorem 10]{Hal94}. 
From these expressions one can deduce that $C_\phi$ is a $G \times G$-invariant map.

(2) Let $w \in W(G,\textbf{T})$ be given. We show that $$C'_\phi (f \circ w) = (C'_\phi f) \circ w.$$ Using the fact that $w \in W(G,\textbf{T})$ acts on $\textbf{T}^\mathbb{C}$ by homomorphisms, we obtain
\begin{align*}
( C'_\phi f )(wt) = \int_{\textbf{T}} f(x) \phi(x^{-1}wt) dx = \int_\textbf{T} f(x) \phi(w (w^{-1}(x^{-1})t)) dx.
\end{align*}
We claim that $\phi(wt) = \phi(t)$ for all $w \in W(G,\textbf{T})$ and $t \in \textbf{T}^\mathbb{C}$. From the invariance of $e^{-2 \pi|Y|^2} \varepsilon$ under $W(G,\textbf{T}$) we see that
\begin{align*}
 \sigma(\pi \circ w) = \int_{\textbf{T}^\mathbb{C}} \|\pi(w\cdot t^{-1}) \|^2 e^{-2 \pi|Y|^2} \varepsilon = \int_{\textbf{T}^\mathbb{C}} \|\pi(t^{-1}) \|^2 e^{-2 \pi|Y|^2} \varepsilon = \sigma(\pi),
\end{align*}
so that
\begin{align*}
 \phi(wt) &= \sum_{\pi \in \hat{\textbf{T}}} \frac{1}{\sqrt{\sigma(\pi)}}\Tr(\pi(w t^{-1})) = \sum_{\pi \in \hat{\textbf{T}}}  \frac{1}{\sqrt{\sigma(\pi \circ w)}}\Tr(\pi(w t^{-1})) \\ &= \sum_{\pi \in \hat{\textbf{T}}} \frac{1}{\sqrt{\sigma(\pi)}}\Tr(\pi(t^{-1})) = \phi(t),
\end{align*}
where in the third step we used the fact that $\pi \mapsto \pi \circ w$ maps $\hat{\textbf{T}}$ bijectively onto itself. Consequently,
\begin{align*}
( C'_\phi f)(wt) = \int_\textbf{T} f(x) \phi( w^{-1}(x^{-1})t) dx.
\end{align*}
The Haar measure on $G$ is invariant under $W(G,\textbf{T})$, because $W(G,\textbf{T})$ acts by automorphisms. Therefore,
\begin{align*}
( C'_\phi f)(wt) = \int_\textbf{T} f(x) \phi(w^{-1}(x^{-1})t) dg = \int_\textbf{T} f(w \cdot g) \phi(g^{-1} t) dg,
\end{align*}
or in other words, $$C'_\phi (f \circ w) = (C'_\phi f) \circ w.$$ Thus  $C'_\phi$ is $W(G,\textbf{T})$-equivariant.
\end{proof}


\end{document}